\documentclass[sigconf,nonacm,10pt]{acmart}

\setcitestyle{nocompress}
\makeatletter
\AtBeginDocument{\let\NAT@space\@empty}
\makeatother

\usepackage{microtype}
\usepackage{booktabs}
\usepackage{tabularx}
\usepackage{array}
\usepackage{stmaryrd}
\usepackage{placeins}
\usepackage{cuted}
\usepackage{mdframed}
\graphicspath{{figures/}}

\makeatletter
\patchcmd{\@mkauthors@iii}
  {\and\par\bigskip}
  {\and\par\bigskip\vspace{6pt}}
  {}{\PackageError{dblog}{Author spacing patch failed}{}}
\makeatother

\AtBeginDocument{%
  \fancypagestyle{standardpagestyle}{\fancyhf{}%
    \fancyfoot[C]{\fontsize{10}{12}\selectfont\thepage}}%
  \fancypagestyle{firstpagestyle}{\fancyhf{}%
    \fancyfoot[C]{\fontsize{10}{12}\selectfont\thepage}}%
  \pagestyle{standardpagestyle}%
  \raggedbottom
  \setlength{\footskip}{30pt}}

\theoremstyle{plain}
\newtheorem{theorem}{Theorem}
\newtheorem{corollary}{Corollary}
\newtheorem{lemma}{Lemma}[section]
\newtheorem{proposition}{Proposition}
\theoremstyle{definition}
\newtheorem{definition}{Definition}[section]
\newtheorem{example}{Example}[section]

\newcommand{\den}[1]{\llbracket #1 \rrbracket}   %
\newcommand{\pre}{\sqsubseteq}                   %
\newcommand{\window}[2]{(#1,#2]}                 %
\newcommand{\keyof}{\mathit{key}}                %
\newcommand{\imgof}{\mathit{img}}                %
\newcommand{\sink}{\mathit{sink}}                %
\newcommand{\keys}{K}                            %
\newcommand{\vals}{V}                            %
\usepackage{ragged2e}
\AtBeginEnvironment{thebibliography}{\RaggedRight\balance}

\title{Generalized DBLog: A Verified Contract for Interleaving
Copied Rows with a Change Log}

\author{Andreas Andreakis}

\begin{document}

\begin{abstract}
Change-data capture (CDC) feeds downstream systems like caches, search
indexes, and data warehouses from a database's log of committed row changes.
When bootstrapping, adding a table, or repairing downstream data, a pipeline
must also copy existing rows. Merging this copy with the active log introduces
the copy-to-log handoff problem. Changes must not fall through a gap, and
older copied state must not overwrite a newer logged update or resurrect a
deleted row. DBLog, developed at Netflix, addressed this problem by reading
tables in chunks and interleaving those reads with the live log. Watermarks
identify the changes that overlap each read, and the log wins when a copied
row is stale~\cite{andreakis2020dblog,dblog_netflix_blog}. Debezium and Flink
CDC have since adapted this
design~\cite{debezium_incremental_snapshots,debezium_mysql_connector,flinkcdc_mysql}.

Earlier work proved that applying the original algorithm's copied rows and
logged changes in their emitted order reconstructs the source's rows, 
including the effect of every logged insert, update, and delete
processed~\cite{andreakis2026virtualcuts}. Generalized DBLog asks when the
same result holds for variants of that design. We state the conditions the
source and capture implementation must satisfy. Once copying and reconciliation 
are complete, we prove that the result holds across all selected tables
and key ranges even when their rows were read at different times. A single
database snapshot is not required for the copy. Further logged changes advance
the reconstructed state one event at a time.

We establish these guarantees for classic watermarking, Debezium's
signal-table and read-only modes, Flink CDC's parallel chunks, reads and dumps
tied to exact log positions, and engine-consistent backups whose log position
lies within known bounds.
The complete theory is machine-checked in Isabelle/HOL, its core independently
verified in Lean~4, and the protocols are also examined by bounded model
checking in TLA\textsuperscript{+}.
\end{abstract}

\keywords{\texorpdfstring{%
  \raggedright
  \mbox{databases}, \mbox{database replication}, \mbox{change-data-capture},
  \mbox{CDC}, \mbox{incremental snapshots}, \mbox{formal verification}, Isabelle/HOL, Lean,
  TLA\textsuperscript{+}%
}{databases, database replication, change-data-capture, CDC,
  incremental snapshots, formal verification, Isabelle/HOL, Lean, TLA+}}

\maketitle

\newpage

\section{Introduction}\label{sec:intro}

\subsection{The production problem}

A change-data-capture (CDC) pipeline typically needs both the latest state
of all rows and an ordered stream of live changes. Because log
retention is finite, the log alone may not supply the needed baseline.
A pipeline must therefore combine a copy of the current source state with
ongoing log consumption.

This need is not limited to an initial bootstrap. An additional sink, a newly added
table, or downstream data loss can require another copy later on. If
only a few keys are damaged, repeating a table-wide or even instance-wide copy
is unnecessary work. These recurring repair and expansion cases make a
one-time bootstrap an incomplete answer for many deployments. Furthermore, the
design problem is not simply how to read current rows. It is how to
combine that read with the change log without losing the effect of an update
or delete, or allowing older copied state to supersede newer logged
state.

\subsection{How copied state joins the change log}

Even a copy obtained from one database snapshot does not,
by itself, identify where a consumer should continue in the change
log. If consumption resumes too late, changes are lost. If it resumes too
early and the overlap is handled in the wrong order, an older copied row can
arrive after a newer logged change and reverse it. A delete that races the
copy can likewise be missed, leaving a row downstream that no longer exists
at the source.

Copying state and handoff correctness are separate
problems. A single snapshot still needs a valid continuation
point. A chunked read can still be correct when its overlap with the log is
properly delimited and reconciled. The mechanisms analyzed here reduce
to two handoff shapes.

\paragraph{Establish an exact handoff point.}
Some databases can associate a read with a position in
their change log. The position states which changes the read already
contains, so the consumer knows where log processing should continue.
A system can also establish the same proof-relevant ordering through
explicit phase coordination that places the copy before the continuing
suffix.
A table read or an entire dump can then be placed at that position
without reconciling an interval of uncertain changes. This is a
strong primitive, but the result depends on the database providing
the promised relationship between the read and the recorded
position~\cite{mariadb_consistent_snapshot,percona_consistent_snapshot,postgres_protocol_replication,mariadb_binlog_status_vars}.
If only an interval containing the true read position is known, the
copy must then instead be reconciled with the changes in that interval.

\paragraph{Bracket and reconcile an overlap.}
A system can keep the change stream moving while it reads current
state. It records positions around the read, treats the interval between
them as an overlap window, and reconciles the copy with the changes in
that window so that the latest logged change wins over a stale copied
row. For example, suppose a read returns key 42 with value V1 and the log
then records an update to V2 before the window closes. Reconciliation
keeps V2. A key with no event in the window keeps its copied value, while
an in-window update or delete drops it.

These approaches are coordination patterns, not mutually exclusive
categories. A whole-table or whole-instance dump can use an
exact position when one is available. If its position is known only to
lie inside an interval, the same dump becomes one large overlapping
read whose racing changes must be reconciled.

An exact handoff relies on its read-coordinate or ordering guarantee. An
overlapping handoff relies on retention, complete observation, and
reconciliation. Scope, repeatability, resumability, unit size, and
parallelism remain separate operational choices. Operators must also
evaluate source load, copy time, log lag, and recovery behavior in
their own environment.

\subsection{From DBLog to Generalized DBLog}

DBLog uses the overlapping approach and makes the overlap small and
repeatable. The original algorithm reads tables in chunks and brackets
each read with two watermark events. Watermarks are changes that are
applied on the source and observed in the change log afterwards. Changes
inside the bracket take
precedence over the corresponding copied rows, so the log wins whenever
it races the copy~\cite{andreakis2020dblog,dblog_netflix_blog}. The design relies on the
watermarks to enclose the unknown read position and make each chunk's
overlap explicit. Figure~\ref{fig:windowing} illustrates one such bracket.
This lets handoff state and progress be managed one
unit at a time. A capture can target all tables, one table, or selected
keys, pause and resume after a failure, and run again later for repair
while change events continue to flow~\cite{andreakis2020dblog}.

\begin{figure}[t]
\begin{mdframed}[
  topline=false,bottomline=false,rightline=false,leftline=true,
  linecolor=black!38,linewidth=1.25pt,
  innerleftmargin=8pt,innerrightmargin=1pt,
  innertopmargin=2pt,innerbottommargin=2pt,
  skipabove=0pt,skipbelow=0pt,nobreak=true]
\small
\RaggedRight
\textbf{Operational impact of watermark writes.}
Writing to the source can initially seem like an unattractive property
of a CDC design. Controlled source-side writes are, however, already an
established observability technique. For example, Debezium can execute a configured
query that writes a heartbeat table when captured data is quiet, and
Oracle GoldenGate updates source heartbeat tables and carries those
records through the replication path to measure end-to-end lag
~\cite{debezium_postgresql_connector,oracle_goldengate_heartbeat}.
A deployment that accepts source-side heartbeat writes has therefore
already accepted the basic operational capability that DBLog
watermarks require, namely small, controlled writes to a dedicated table whose
events are expected to pass through the capture pipeline. Heartbeats use
that capability to provide a liveness signal. DBLog watermarks reuse it
to delimit chunk reads.
\end{mdframed}
\Description{A left-rule sidebar explaining that source-side heartbeat
writes make DBLog watermark writes an operationally familiar CDC
interaction, while read-only bracket mechanisms remain available.}
\end{figure}

\begin{figure*}[t]
\centering
\includegraphics{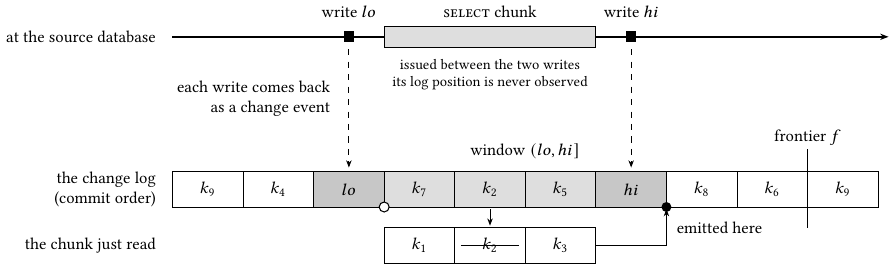}
\caption{DBLog's windowing. Each
watermark is an update to a dedicated single-row table at the source, so
it comes back from the log as an ordinary change event. Those two events
are the window's edges. The window is half-open, with the $lo$ cell lying
outside it and the $hi$ cell inside. The chunk read is issued between the
two writes. In this example, its own log position, which is never observed, lies
inside the window. A copied row is drawn in the column of the logged
change that shares its key, not at a log position of its own. The logged
$k_2$ has already gone out, so the copied $k_2$ is discarded, and only
the survivors are emitted when $hi$ arrives.}
\label{fig:windowing}
\Description{Three registered rows over one shared column grid. The top
row is a real-time axis of what the capture does at the source database,
showing the low watermark write, chunk select, and high watermark write.
The middle row is the change log drawn as an array of cells, each cell
holding the key that event changed, with the two watermark events as
shaded cells named lo and hi and the cells between them shaded as the
window. Dashed vertical arrows run from each source write down to its
watermark cell, and a hatched bar between the two marks the read's
unobserved log position. The bottom row is the chunk just read, drawn on
the same columns. The copied row sharing a key with a logged change
inside the window is struck through, and the survivors are routed back
into the log immediately after the high watermark.}
\end{figure*}

Earlier work proved that replay of this original watermarked algorithm
matches the source at a certain log position, representing all changes
up to that point~\cite{andreakis2026virtualcuts}. Open-source CDC projects have
since adopted and adapted the design. Debezium's documented default
incremental-snapshot mode writes opening and closing window records to a
configured signaling data collection under its
\texttt{insert\_insert} strategy. Its read-only implementation for MySQL and
MariaDB instead samples server GTID state without writing those markers
~\cite{debezium_signalling,debezium_mysql_connector,
debezium_mariadb_connector}. For PostgreSQL, it samples the server's transaction snapshot before and after each chunk read~\cite{debezium_ddd8,debezium_postgresql_connector}.
Flink CDC's documented MySQL protocol brackets chunks in parallel and
reconciles each chunk with the changes that raced its read. Its documentation
records that the design is inspired by DBLog~\cite{flinkcdc_mysql}. We
also study log-coordinate-bound reads and native dumps, which realize
the exact-position or uncertain-interval approaches described
above~\cite{mysql_mysqldump}.

These variations motivate the paper's central question. Under what
conditions can copied state be merged with the continuing log so that
replay neither loses source changes nor lets older copied state
overwrite newer logged state? Generalized DBLog provides a framework
for answering that question. Under this framework, replaying the combined stream reconstructs the
source state at a single log position, representing the source's state 
history up to that point. We call this a \emph{virtual cut}. Importantly, 
this cut is not stored inside DBLog. Instead, it is the state represented 
by replaying DBLog's output. Unlike a one-time physical snapshot
that represents all rows at one common source coordinate, a virtual cut is assembled
iteratively and on an ongoing basis by interleaving chunk reads with log changes.

The following sections establish what each capture protocol must guarantee to achieve this result.

\subsection{Contributions and reading paths}

\paragraph{Contributions.}
We make four contributions:

\begin{enumerate}
\item We state a correctness contract for combining copied database
state with a committed change log and prove its cut theorem. Replaying
the combined stream matches the source for the keys processed so far
(Lemma~\ref{lem:closed}), and covers all target tables and keys
once all chunks finish (Theorem~\ref{thm:cut}). The correctness
condition is per key and requires no global database snapshot. Chunks
can be read at different times, provided each read reflects source state
somewhere within its watermark bracket. From the final watermark onward,
the reconstructed state continues to track the live source as subsequent
change events arrive.

\item We prove that the original watermarked algorithm satisfies the contract,
and show that other CDC designs fulfill it as well.
This covers Debezium's signal-table modes (\texttt{insert\_insert} and
\texttt{insert\_delete}), the Debezium 3.6 MySQL, MariaDB, and PostgreSQL
read-only variants, the Flink CDC 3.6 MySQL parallel-chunk protocol,
log-coordinate-bound reads, and native dumps, establishing that each
correctly reconstructs source state. On a PostgreSQL primary, the
watermark-based and read-only results need further conditions
(\S\ref{sec:inst-classic}, \S\ref{sec:inst-readonly}).

Our analysis is based on version-pinned documentation and selected
source code identified in \S\ref{sec:instances}. The executed-GTID proof
establishes that sampled GTID sets reliably bound each chunk's read
window when accepted samples represent exact, transaction-complete prefixes
of the same history. The parallel proof demonstrates that chunks over
disjoint keys require no relative ordering or coordination between their
brackets.

\item We separate correctness criteria from operational considerations. 
For correctness, the log must describe one complete history, 
the copy must partition target keys into non-overlapping chunks, 
and reconciliation must let later logged changes override stale copied rows. 
Operational choices (like chunk size, parallelism, bracket width,
and whether watermarks require source writes) can affect performance and
throughput, but do not alter the underlying correctness result.

\item We verify the framework using both interactive theorem provers and model checking. 
The full development is machine-checked in Isabelle/HOL, an independent Lean~4 
formalization checks the central cut theorems, and five 
bounded TLA\textsuperscript{+} models exercise the operational 
protocols and their composition.
\end{enumerate}

This formulation generalizes the virtual cut introduced in earlier
work on the classic algorithm~\cite{andreakis2026virtualcuts}.
The theory in this paper is self-contained. The mechanizations are
independent checks, not prerequisites for reading the proofs.

\paragraph{Reading paths.}
Practitioners can go directly from this introduction to Section~\ref{sec:practitioner}. It
presents Generalized DBLog in plain operational terms, separating what the
database must provide from what the connector must establish.
In addition, Table~\ref{tab:instances} summarizes how production
protocols, including Debezium and Flink CDC, fit into this framework.

Readers interested in the formal theory should read
\S\S\ref{sec:setting}--\ref{sec:instances} in order for the system model,
the virtual cut theorems, the reconciliation algorithms, and the connector
proofs, followed by \S\ref{sec:mech} for the Isabelle/HOL, Lean~4, and
TLA\textsuperscript{+} mechanizations. In addition,
Section~\ref{sec:audit} revisits the original 2020 DBLog
paper~\cite{andreakis2020dblog} to make its datastore
assumptions explicit, \S\ref{sec:related} discusses related work, and
\S\ref{sec:fence} outlines the scope and limitations of the work.

\clearpage
\section{The Setting: Stores, Logs, and Coordinates}\label{sec:setting}

We model one \emph{capture source} as a keyed store together with a
commit-ordered change log. All definitions are per capture scope,
the set of tables or key ranges that one capture run is responsible
for.

In the terminology used below, each copied chunk is a \emph{unit},
its two log positions form a \emph{bracket}, and every copied row or
observed absence becomes a \emph{read result}. The
\emph{frontier} is the position at which we judge the combined output.

The contract answers five questions: (1) which complete, commit-ordered
row history is authoritative, (2) which stable row identities the copy
must account for, (3) which low and high coordinates bracket each unit's
read, (4) whether replaying the actual emission matches letting the log override
the copy, and (5) through which observed coordinate the result is claimed.

The contract alone fixes the canonical replay at the frontier. Emitted
streams, transaction closure, and shared witnesses are added in later
sections, while physical prefixes, restart, transport, and sink application
stay outside the result.

\subsection{Keyed stores and commit-ordered logs}\label{sec:setting-store}

Fix a key space $\keys$ and a value domain $\vals$. A \emph{state}
is a partial map $\sigma : \keys \rightharpoonup \vals$, where
$\sigma(k) = \perp$ means key $k$ is absent. For a multi-table
scope, $\keys$ is the disjoint union of per-table key spaces, each
key tagged with its table, and nothing in the theory depends on the
tagging.

The elements of $\keys$ are \emph{stable logical row identities}, not
simply values used to order a scan. A primary-key change is normalized
as absence for the old identity followed by a value for the new one.
A table without a primary key needs another stable identity whose
equality is shared by the read and log paths. Ambiguous duplicate rows
do not fit the keyed model without an additional normalization that
supplies such identities.

The committed history is a finite or growing sequence of
\emph{event occurrences} $e_1, e_2, \ldots$ in \emph{commit order}.
Position supplies occurrence identity. Two occurrences remain distinct
even when their keys and images are equal, and we suppress that position
tag in the notation. An event occurrence carries a key
$\keyof(e) \in \keys$ and determines the post-state of
that key, $\imgof(e) \in \vals \cup \{\perp\}$, with $\perp$ for a
delete. This is the row-event shape a log-based CDC deployment
configured for full row images consumes, consisting of a binlog or WAL entry
carrying the complete after-image
of one row, or a tombstone. Events outside the scope are ignored
throughout. Logs that carry deltas rather than full images are 
addressed by S-IMG in \S\ref{sec:contract-promises}.

Only committed transactions contribute events to this history. An
event from an active or aborted transaction is not in the modeled log,
cannot be selected by a coordinate or window, and cannot participate in
replay. Source reads that supply copied values likewise observe
committed database state. Dirty reads are outside the contract.

State and log are tied together by replay. Writing $P$ for a log
prefix, the \emph{row-event replay state} after $P$ is the per-key
fold $\sigma_P$:
$\sigma_\emptyset$ is the initial state, and $\sigma_P(k)$ is
$\imgof(e)$ for the last event $e \in P$ with $\keyof(e) = k$, or
$\sigma_\emptyset(k)$ if no such event exists. Every committed write
appears as exactly one event per changed row, so $\sigma_P$ is the
authoritative replay state of those row events at ``time'' $P$.
Times, here and everywhere below, are log prefixes, never wall
clocks. The core deliberately carries no transaction identifier, so
a mathematical row-event prefix may end between adjacent events of a
multi-row transaction that has already committed and entered the log as
one complete block. Such a prefix is only an intermediate state of
replaying already-committed row events. It does not represent a partial
commit or access to an in-progress transaction. We reserve
\emph{transactional snapshot} and
\emph{committed database state} for a transaction-closed prefix, one
containing either all or none of each transaction's row events. Any
such reading of a theorem below therefore needs transaction closure
as an additional instance assumption. The unconditional result is
per-key equality with the row-event replay state.

\subsection{Mapping coordinates to log prefixes}\label{sec:setting-coords}

Real systems do not pass log prefixes around. They pass
\emph{coordinates}, such as binlog file/position pairs, LSNs, SCNs,
executed-GTID sets, exported snapshot names, and resume tokens. The
contract requires only that a coordinate mark a
position whose meaning is ``these committed row events, and no later
ones, lie before here.'' A transactional-snapshot claim additionally
requires that position to be transaction-closed.

\begin{definition}[Coordinate space]\label{def:coords}
A \emph{coordinate space} is a set $C$ together with a mapping
$\den{\cdot} : C \to \mathrm{Prefixes}(\mathrm{Log})$ such that all
prefixes in its image are linearly nested, meaning that for any $c, c'$, either
$\den{c} \subseteq \den{c'}$ or $\den{c'} \subseteq \den{c}$. Write
$c \pre c'$ for $\den{c} \subseteq \den{c'}$. Thus $\pre$ is a total
preorder on $C$, and quotienting by equality of mapped prefixes gives a linear
order.
\end{definition}

Operationally, a usable coordinate lets the capture faithfully answer whether a committed event occurred between two
recorded positions.

The only operation the contract ever needs is the \emph{membership
oracle}, which, given an event occurrence $e$ and coordinates $c \pre c'$, decides
whether $e \in \den{c'} \setminus \den{c}$, that is, whether $e$
lies inside the window $\window{c}{c'}$.
No scalar ordering or interval arithmetic over raw coordinate
values is required. In practice, scalar LSNs and file/position pairs map to prefixes
through their usual comparisons, but executed-GTID \emph{sets}
map to prefixes too, and their membership oracle is set membership
on the difference of two GTID sets, with no interval arithmetic and
no scalar anywhere. This modeling choice is what makes read-only capture an ordinary member of the family
(\S\ref{sec:instances}).

Admissibility is nestedness plus \emph{operational fidelity}: the 
instance's window-membership decisions must agree with the declared
$\den{\cdot}$. Raw coordinate comparisons used to establish plan 
bounds must imply prefix inclusion. Distinct raw coordinates may still 
map to the same prefix.

A mapping that ignores the system's actual coordinate behavior
proves nothing. Mapping every coordinate to the empty prefix is nested but useless, because a plan
that consults coordinates then disagrees with its own
declared reading. And a system whose operational coordinates admit
\emph{no} mapping to commit-order prefixes that matches their behavior (client-supplied
timestamps, sequence numbers allocated at statement time rather
than commit time, replicas applying transactions out of commit
order) is outside the contract altogether. That is level T0
(\S\ref{sec:tiers}), the family boundary. Membership is not a matter of degree. Either the system's coordinates can be read as log prefixes or they cannot. Nestedness itself costs nothing,
because mapped prefixes belong to one linear log, and prefixes of a
single sequence are always nested. Fidelity is therefore the only
substantive admissibility requirement, verified instance by
instance (\S\ref{sec:instances}).

\subsection{Windows and the cutoff discipline}\label{sec:setting-windows}

\begin{definition}[Window]\label{def:window}
For $lo \pre hi$, the \emph{window} $\window{lo}{hi}$ is the ordered
event-occurrence slice written $\den{hi} \setminus \den{lo}$.
Occurrences inherit commit order from the log. Equal key-and-image
payloads at different positions remain distinct elements.
\end{definition}

Windows are half-open by convention, meaning events at the low edge are
excluded and events at the high edge included. Every instance fixes its
edges to this convention once. Including or excluding one boundary
event twice would duplicate it. Excluding it from both adjacent
windows would lose it. The same off-by-one failure can occur at
$lo$ or $hi$, so the convention is part of the formal object. An
instance whose raw endpoint records use a different syntax (say, an
offset-inclusive $[\mathit{LOW}, \mathit{HIGH}]$ span in a vendor
document) must declare its endpoint-to-prefix normalization as an
instantiation obligation before its plan may use the window. The
semantic window is always $\den{hi} \setminus \den{lo}$.

\subsection{What stays outside}\label{sec:setting-fence}

The contract is source-side. The replay semantics defined in
\S\ref{sec:contract} is a mathematical fold over the emitted
sequence. It defines what the emitted stream \emph{means}, so that
``cut'' is a well-defined claim about capture. It is not a
delivery-semantics claim. Transport, exactly-once application, sink
transactionality, schema and DDL evolution, deployment topology,
and security stay outside this paper's scope and are cited, not
treated (\S\ref{sec:fence}). The delivery side of the boundary is
the subject of separate work~\cite{andreakis2026dualwrite}. This boundary has one declared exception, namely the \emph{scope witness}.
Determining which relations and keys exist is itself a read that
needs a coordinate. Every deployment must therefore satisfy the common
\emph{A-scope-witness} condition by naming how the relations and key ranges
were enumerated and which coordinate binds that enumeration. The
production instances share this common condition. Their individual
specifications add only mechanism-specific scope facts. Schema 
evolution during capture remains out of scope by definition.

\section{Capture Plans and the Contract}\label{sec:contract}

A capture run has three parts relevant to the proof. It divides the
scope into units, reads each unit within a log bracket, and uses an
associated merge discipline to emit reads together with log events.
The \emph{capture plan} records the final units, brackets, read
results, and frontier. It does not record the emitted stream.
Scheduling and emission enter separately through the instance and
merge obligations.

\subsection{Capture plans}\label{sec:contract-plans}

\begin{definition}[Capture plan]\label{def:plan}
A capture plan for scope $\keys$ at frontier $f \in C$ is a finite
family
\[
\mathrm{Plan} \;=\; \{\, (D_i,\; lo_i,\; hi_i,\; R_i) \;:\; i = 1..n \,\}
\]
where $\{D_i\}$ partitions $\keys$. Each unit carries a
\emph{bracket} $(lo_i, hi_i)$ with
$lo_i \pre hi_i \pre f$. The map
$R_i : D_i \to \vals \cup \{\perp\}$ is
the unit's \emph{refresh}, assigning one value-or-absence read result per key of the unit.
\end{definition}

Operationally, $\keys$ is the captured identity namespace and each
$D_i$ is the final ownership predicate for one unit, often a table and
half-open key range. These sets are not assumed finite. For a range or
table capture they include identities that are currently absent and
identities whose rows may be inserted while the capture runs. The
finite family of final domains must partition the scope. A finite scan
still induces the total map $R_i$, where returned rows give values
and a nonreturned identity indicates implicit absence only when
the scan and range predicate establish complete membership. A
concurrent insert already belongs to one final domain. Its early
inferred absence is valid before the insert, and the observed insert
event later overrides it through the log-wins rule.

Brackets may be shared between units (one bracket amortized over
many reads), degenerate ($lo_i = hi_i$, a \emph{point bracket}), or
global ($lo_i$ at capture start). The frontier $f$ is the coordinate
up to which this run makes any claim.
Figure~\ref{fig:anatomy} draws the anatomy of one unit.

\begin{figure}[t]
\centering
\includegraphics{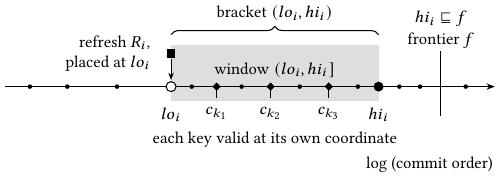}
\caption{One unit of a capture plan. The bracket $(lo_i,hi_i)$
spans a segment of the commit-ordered log. The window $(lo_i,hi_i]$
is half-open, with events at the low edge excluded and events at the
high edge included. The canonical rule places the refresh $R_i$ at
$lo_i$, wherever its rows were actually read. Condition O2 
(bracket-local validity) requires less than a single read instant. Each 
key of the unit reflects the database state at some coordinate
of its own inside the bracket
($c_{k_1}$, $c_{k_2}$, $c_{k_3}$). Every bracket sits at or below
the frontier $f$.}
\label{fig:anatomy}
\Description{A horizontal commit-order log axis with a unit's
bracket marked by an open circle at the low edge and a filled
circle at the high edge, the half-open window shaded between them,
the refresh square placed at the low edge, three per-key read
coordinates at distinct positions inside the bracket, and the
frontier to the right of the bracket.}
\end{figure}

The definition is silent about \emph{how} $R_i$ was obtained, whether by a
chunked \texttt{SELECT}, a \texttt{COPY}, a dump file, a restored
physical backup, or a read against a replica. It is also silent about
scheduling (how large units are, and whether they ran serially or in
parallel) and about how brackets were realized. It carries their
endpoints, but imposes no constraint on their numerical width. Scheduling
and width do not change the theorem once the contract, observation,
and emission-equivalence conditions continue to hold, though either can
make those requirements operationally infeasible. Provenance likewise has
no field in the tuple, but it can determine which O4 and O5 obligations
apply and therefore whether contract membership is established. The
plan encodes no physical simultaneity. Corollary~1 instead tests the
extensional property that one shared state witness lies in every bracket
and agrees with every refresh. An engine snapshot is one way, not the
only way, to realize that fact (\S\ref{sec:theorem}).

Real runs are dynamic. They re-split large chunks mid-flight and
discover tables as they go. The definition absorbs this by reading
$\mathrm{Plan}$ as the \emph{final} partition the run produced. A
run that re-splits after it has already emitted reads must satisfy the
corresponding equivalence obligation
(\S\ref{sec:contract-replay}) with respect to the final plan. The
actually-emitted stream, including any superseded pre-split
refreshes, must replay to the final plan's sink state.

\subsection{What the source promises}\label{sec:contract-promises}

Three promises bind the source and its log. They carry the 2020 paper's requirements forward, kept in spirit and sharpened where \S\ref{sec:audit} found them under-specified.

\begin{itemize}
\item \textbf{(S-LOG: commit-ordered keyed log.)} The log is a
linear commit-order event sequence over $\keys$. Every committed
change to a scoped key appears as exactly one event, and no event from
an active or aborted transaction appears. A committed transaction's
row events become available as one complete block. Key identity is preserved.
\item \textbf{(S-IMG: image sufficiency.)} $\imgof(e)$ fully
determines the post-state of $\keyof(e)$. Where a real log carries
deltas, or primary-key updates without the old key, the instance
must either strengthen the log (before-images, delete-plus-insert
rendering, replica-identity-class settings) or adopt a re-read
repair \emph{as a log normalization}. The repaired instance is
modeled over effective full-image events, each re-read supplying the
image its delta event lacked, and its equivalence obligation then
includes that normalization. Absent strengthening or normalization,
delta logs are outside this core.
\item \textbf{(S-OBS: observability.)} The capture process
observes every scoped event of $\den{f}$ from its consumption
start $s_0$ onward, requiring both log coverage and
the retention to read what its merge discipline consumes. Where $s_0$ must sit, and
in what access mode window ranges must be readable, is owned by
the discipline in use, and each discipline's definition states its
own bound (\S\ref{sec:merge}). A single pass from $s_0 \pre lo_i$
for every unit suffices for window-discard, while range-merge
streams from $s_0 \pre hi_i$ and reads windows by
coordinate-addressable access.
\end{itemize}

\subsection{Capture plan requirements}\label{sec:contract-obligations}

\begin{itemize}
\item \textbf{(O1: coverage.)} $\{D_i\}$ partitions $\keys$,
so every scoped key is owned by exactly one unit. Plans whose units
overlap are outside the core contract. A composition with
overlapping branches must satisfy the coherence obligation that the
branches replay identically, and this paper treats only partitions.
\item \textbf{(O2: bracket-local validity.)} For every unit $i$
and every key $k \in D_i$ there \emph{exists} a coordinate $c$ with
$lo_i \pre c \pre hi_i$ such that
$R_i(k) = \sigma_{\den{c}}(k)$.
\item \textbf{(O3: domain completeness.)} $R_i$ is total on
$D_i$, so that every key of the unit yields exactly one value-or-absence
result.
\item \textbf{(O4: declared external assertions.)} Where a
bracket coordinate is asserted by tooling rather than constructed or
observed by the capture process (dump metadata, a backup's recorded
recovery point), the coincidence ``$R_i$ was read at the asserted
coordinate'' enters as an \emph{explicit assumption} of the instance,
never as a proved fact.
\item \textbf{(O5: representation agreement.)} The refresh path
and the log path agree on key identity and value encoding.
Formally, both land in the same $\keys$ and $\vals$. Operationally,
a checkable clause that applies exactly when the refresh bypasses the
connector's row path, as dumps and backups do.
\end{itemize}

O2 is the central weakening of this paper. It is
\emph{per-key} and \emph{existential}.
Different keys of one unit may reflect the source at different in-bracket
coordinates, and the contract does not require one shared read
coordinate. Many concrete reads satisfy O2 more strongly, as when a statement
snapshot or coordinate-bound read supplies the same read coordinate for every
key in its unit. The classic algorithm's single-coordinate chunk reads
satisfy O2 in this way (\S\ref{sec:instances}). On a PostgreSQL primary,
however, the transactions in a statement snapshot do not always form a
prefix of the log (\S\ref{sec:inst-classic}). The per-key
form also admits a complete refresh whose keys reflect different read coordinates.
Section~\ref{sec:theorem-torn} constructs such an instance
and proves that no single coordinate explains all of its read results,
while its cut still holds.

The database operation that supplies $R_i$ is a committed read. When an
instance retains transaction labels, its operational read coordinate lies
at a transaction boundary. The executed-set instance of
\S\ref{sec:inst-readonly} makes this explicit. The transaction-free core
records only the per-key equality required by the proof. It does not
turn an interior row-event replay prefix into an allowable dirty-read
state.

Together, O2 and O3 pin window-invariant keys exactly. If no
in-window event touches $k$, then all candidate values
$\sigma_{\den{c}}(k)$ across the bracket coincide, so the read
result $R_i(k)$ \emph{is} that invariant value.

A scan that silently misses a key cannot satisfy both obligations,
because absence-by-omission is sound only when the key is genuinely absent.
This pair, not read freshness and not a specific isolation level, 
is the property a scan must preserve. Section~\ref{sec:audit} returns to
the point, because the 2020 paper's own requirement was stated in
freshness vocabulary and the gap between the two formulations is where
one class of engines lacks statement-stable scan membership.

O4 makes external trust explicit. Splice-style instances use
coordinates that arrive from outside the capture process,
such as a coordinate written into a header by a dump tool or a recovery point
recorded by a backup system. The contract does not treat these as observations. 
Each such reliance becomes an explicit assumption 
that the instance records visibly (\S\ref{sec:instances}), 
so that every level or certification claim in this paper is
readable as ``proved, given this list.'' Where the list is empty,
no additional external coordinate assertion is assumed.
Otherwise the reader can check each entry against their own deployment.

\begin{definition}[Source-side contract]\label{def:contract}
A source, a coordinate space, and a capture plan \emph{satisfy the
contract} iff S-LOG and S-IMG hold and the plan satisfies
O1-O5.
\end{definition}

S-OBS stands outside the definition for a reason. Within the model, 
S-LOG and S-IMG specify the commit-ordered full-image event sequence 
of \S\ref{sec:setting}. An instantiation over a real system carries 
them as that modeling claim. The
contract and its cut theorem constrain the \emph{canonical}
replay, a fold over the committed history directly
(\S\ref{sec:contract-replay}), and require no observation clause.
What a capture process actually observed becomes decisive at the point where an emitted stream is claimed
equivalent to the canonical one, and there each discipline's
equivalence theorem relies on S-OBS in the discipline's own form
(\S\ref{sec:merge}). A deployment must satisfy all three source requirements. The
split records which theorems depend on which clauses.

\subsection{Canonical replay and emission equivalence}\label{sec:contract-replay}

For canonical \emph{replay}, we place each unit's refresh at $lo_i$ and 
let subsequent log events override it. At the frontier, this gives 
the per-key fold: for $k \in D_i$,
\[
\sink_f(k) \;=\;
\begin{cases}
\imgof(e^{\ast}) & \parbox[t]{0.52\columnwidth}{if $e^{\ast}$, the
  last event for $k$ in $\window{lo_i}{f}$, exists,}\\[2pt]
R_i(k) & \text{otherwise.}
\end{cases}
\]
Placing refreshes at $lo_i$ and letting later events win is the
\emph{canonical merge rule}. The refresh overrides events at or 
before $lo_i$, and events in $(lo_i,f]$ override the refresh.
The low edge is a conservative placement.
Condition O2 permits each read result to reflect the source state at some coordinate within the bracket,
while the fold gives every later in-window and post-window
event final precedence. It is a definition of meaning, not an
implementation prescription, and shipped systems do not need to emit
this stream. How far a real discipline departs from it varies. The
classic algorithm buffers
a chunk and \emph{discards} buffered rows as in-window events
arrive. Flink-shaped systems \emph{merge} the window's events onto
the chunk and emit the unit once. Splices emit a dump behind a
point bracket with nothing to reconcile. Their modeled
copy-before-suffix stream is replay-equivalent to the canonical rule
when that order, or an equivalent merge, is preserved. Each such
discipline $M$ is admitted through an \emph{equivalence
obligation}. The instance must show that $M$'s emitted stream,
built from what the capture observes under $M$'s own S-OBS form,
replays to the same $\sink_f$ as the canonical rule. Usually that
is a short proof, and \S\ref{sec:merge} is a catalog of the
established ones, alongside entries whose obligations are stated but
left open, labeled as that. This proof approach makes ``the log wins'' a theorem about the canonical rule rather than a separate argument per implementation. Several industrial correctness conditions, such as withholding rules and buffer-close semantics, surface at this layer (\S\ref{sec:merge}).

The next section establishes the contract guarantee. For
every plan satisfying Definition~\ref{def:contract}, the canonical
replay $\sink_f$ is exactly the source's row-event replay state
$\sigma_{\den{f}}$, on every key, at the frontier, regardless of
provenance, parallelism, or width.

\section{The Cut Theorem and Its Corollaries}\label{sec:theorem}

The contract yields one theorem, a derived latest-high lemma, a
closed-unit lemma, and three corollaries. The
theorem states that a conforming plan's canonical replay equals the
source's row-event replay state at its frontier. The corollaries carry
the readings a practitioner needs beyond the frontier value itself.
They show how an exhibited shared state witness can move the canonical
trajectory earlier (Corollary~\ref{cor:shared}), what a capture may claim past the point
where it stopped observing (Corollary~\ref{cor:cont}), and why plans
assembled from different capture methods compose into one guarantee
(Corollary~\ref{cor:mix}). The theorem's proof and the first two corollaries' are local to
one key and its one covering unit, and the third corollary composes whole
plans without ever coupling their units. That locality is itself a result, 
and Section~\ref{sec:instances} relies on it.

\subsection{Window invariance and the theorem}\label{sec:theorem-cut}

The entire development rests on one lemma. It says that if nothing
happened to a key inside a window, then the key's value is the same
at every coordinate of that window. Hence a read that matches 
the source state at any coordinate in such a window matches it at every coordinate in that window.
Bracket width therefore never enters the canonical cut theorem.
Widening a bracket enlarges the set of events the merge must
reconcile, and can strain its observation requirements, but for the keys
the refresh actually decides, those with no event in the window,
every coordinate of the bracket is as good as every other.

\begin{lemma}[Window invariance]\label{lem:inv}
Let $lo \pre c \pre hi$ and let $k$ be a key with no event in the
window $\window{lo}{hi}$. Then
$\sigma_{\den{c}}(k) = \sigma_{\den{lo}}(k)$. Consequently
$\sigma_{\den{c}}(k) = \sigma_{\den{c'}}(k)$ for any two coordinates
$c, c'$ of the bracket.
\end{lemma}

\begin{proof}
The prefix $\den{c}$ extends $\den{lo}$ by exactly the window
$\window{lo}{c}$, and $\window{lo}{c} \subseteq \window{lo}{hi}$
because $c \pre hi$. The extension therefore contains no event for
$k$, and the per-key fold of \S\ref{sec:setting-store} leaves
$\sigma(k)$ unchanged across an extension without events for $k$. The
consequence follows by applying this to $c$ and to $c'$ and comparing
both against $lo$.
\end{proof}

\begin{theorem}[Cut theorem]\label{thm:cut}
Let a source, a coordinate space, and a capture plan at frontier $f$
satisfy the contract (Definition~\ref{def:contract}). Then the
canonical replay equals the source's row-event replay state at the
frontier on every key of the scope:
\[
\sink_f(k) \;=\; \sigma_{\den{f}}(k)
\qquad\text{for all } k \in \keys.
\]
\end{theorem}

\begin{proof}
Fix $k \in \keys$. By O1 exactly one unit $i$ owns $k$, and both
sides of the claim are decided by the events for $k$ alone. The left
side consults the window $\window{lo_i}{f}$ and the refresh
$R_i(k)$, while the right side is the image of the last event for
$k$ in $\den{f}$, or $\sigma_\emptyset(k)$ if there is none. Since
$lo_i \pre hi_i \pre f$, the prefix $\den{f}$ consists of
$\den{lo_i}$ followed, in commit order, by the window
$\window{lo_i}{f}$. The events for $k$ in $\den{f}$ split the same
way. We distinguish whether the window part is empty.

\emph{Case A: some event for $k$ lies in $\window{lo_i}{f}$.} Let
$e^{\ast}$ be the last such event. The replay returns
$\sink_f(k) = \imgof(e^{\ast})$ by definition. By the split,
$e^{\ast}$ is also the last event for $k$ in all of $\den{f}$, so
$\sigma_{\den{f}}(k) = \imgof(e^{\ast})$ as well. The fold discards
the refresh in favor of $e^{\ast}$ regardless of what the read
returned for $k$. Condition O2 for 
$R_i(k)$ is not needed here, only its existence (O3), and the
eviction behavior of every real merge discipline
(\S\ref{sec:merge}) implements this property of the fold.

\emph{Case B: no event for $k$ lies in $\window{lo_i}{f}$.} The
replay returns $\sink_f(k) = R_i(k)$, and by O3 that read exists
and is unique. O2 supplies a coordinate $c$ with
$lo_i \pre c \pre hi_i$ and $R_i(k) = \sigma_{\den{c}}(k)$. The window $\window{lo_i}{f}$ contains no event for $k$, and $lo_i \pre c \pre f$ since $c \pre hi_i \pre f$. Two applications
of Lemma~\ref{lem:inv}, with $f$ as the window's high edge,
therefore give
\[
R_i(k) \;=\; \sigma_{\den{c}}(k)
       \;=\; \sigma_{\den{lo_i}}(k)
       \;=\; \sigma_{\den{f}}(k),
\]
which is the claim.
\end{proof}

Two remarks on the proof's shape. Both are used later.

\emph{The frontier is parametric.} Beyond the bracket bounds
$lo_i \pre hi_i$, the argument relied on exactly one property of $f$,
namely that $hi_i \pre f$ for the covering unit. The identical argument
therefore proves replay to \emph{any} coordinate $g$ bounding every
unit's high edge equal to $\sigma_{\den{g}}$. This parametric form
is what Lemma~\ref{lem:latest-high} and
Corollary~\ref{cor:cont} instantiate.

\emph{Nothing couples distinct units.} The proof fixes one key,
finds its one unit, and never mentions another. No condition couples 
two units' brackets, no step compares their read coordinates, and no
bound on how many units run concurrently appears anywhere. Nothing
couples distinct units beyond O1 disjointness and the shared frontier
bound. Whether the units ran serially or in parallel, and whether
their brackets were disjoint, nested, or identical on the log axis,
is invisible to the argument. Corollary~\ref{cor:mix} applies this
observation, and the parallel-chunks instance of
Section~\ref{sec:instances} derives its guarantee from it without any
coordination argument of its own.

\paragraph{Frontier equality is per-key.}
The equality of Theorem~\ref{thm:cut} holds key by key at the
frontier. Each replayed value equals that key's value in the
row-event replay state at the frontier. The core framework assumes 
nothing about transaction boundaries. Consequently, calling that state a 
transactional snapshot additionally requires the frontier to be
transaction-closed. Corollary~\ref{cor:shared}'s trajectory requires
the same closure at every coordinate given that reading. The
contract alone supplies neither closure nor physical sink-prefix
behavior. Section~\ref{sec:practitioner-frontier} gives a concrete
practitioner example of the distinction between committed-only input
and a transaction-closed result. Where coordinates are transaction
boundaries, as with the
executed-transaction sets of \S\ref{sec:inst-readonly}, that stronger 
structure enters through instance-specific assumptions.

\paragraph{Monotonicity is discipline-scoped.}
Theorem~\ref{thm:cut} is a statement about the final replayed state.
Whether the sink also moves \emph{monotonically} while the emission
is consumed, never showing an older value for a key after a newer
one, depends on the merge discipline and is proved per discipline in
\S\ref{sec:merge}. Under the canonical rule itself, monotonicity
holds for every window-invariant key. A key touched inside
its window can step backward transiently. Its refresh is placed at
$lo_i$ even when the read that produced it happened late in the
bracket, so replay may apply that late value first and the key's
earlier in-window events after it, until the last event restores
the final value. The transient is confined to the window that
produced it, and the final value is unaffected. The two production
disciplines both eliminate it. Window-discard emits a key's events
in commit order with the surviving refresh placed at the window
close, and range-merge emits the merged state for each key at the close
followed only by later events. Both per-key monotonicity theorems
are stated and proved in \S\ref{sec:merge}.

\subsection{A worked replay}\label{sec:theorem-example}

Executing the theorem's fold once by hand shows both the log winning 
and what the refresh is actually for.

\begin{example}[Window replay]\label{ex:replay}
Take three keys $a, b, c$ over an initially empty store
($\sigma_\emptyset$ assigns $\perp$ everywhere), and a committed
history of three events
\[
L = e_1\, e_2\, e_3,
\qquad
\begin{aligned}
e_1 &: \keyof(e_1) = a,\; \imgof(e_1) = 1,\\
e_2 &: \keyof(e_2) = a,\; \imgof(e_2) = 2,\\
e_3 &: \keyof(e_3) = c,\; \imgof(e_3) = 9.
\end{aligned}
\]
Coordinates $c_0, c_1, c_2, c_3$ map to the prefixes of lengths $0$
through $3$. The plan has one unit, with domain $\{a, b, c\}$, bracket
$(lo, hi) = (c_1, c_2)$, frontier $f = c_3$, and refresh
\[
R(a) = 1, \qquad R(b) = \perp, \qquad R(c) = \perp,
\]
a read that matches the source state at the low edge 
$c_1$. After $e_1$, key $a$
held $1$ and keys $b, c$ were absent, so O2 is satisfied with
witness $c_1$ for all three keys, and O3 holds since $R$ gives a
read result for each. The replay to $f$ works each key through the
window $\window{c_1}{c_3} = e_2\, e_3$:
\begin{itemize}
\item \emph{Key $a$ (a stale read, evicted).} The window contains
$e_2$, so Case~A applies, giving $\sink_f(a) = \imgof(e_2) = 2$. The
refresh value $1$ was true when the cursor passed $a$ and is stale
at the frontier. The fold never consults it. This is ``the log wins.''
\item \emph{Key $b$ (true absence, kept).} No event for $b$ exists
anywhere, so Case~B applies and $\sink_f(b) = R(b) = \perp$, which
is correct because O2 forces the recorded absence to be
valid somewhere in the bracket and Lemma~\ref{lem:inv} extends it
to $f$.
\item \emph{Key $c$ (superseded after the bracket).} The window
$\window{c_1}{c_3}$ contains $e_3$, which committed after the high
edge $c_2$ but at or before the frontier. Case~A applies,
yielding $\sink_f(c) = 9$. The refresh recorded $c$ as absent, validly inside the
bracket, yet the later event supersedes it, as a live
stream should.
\end{itemize}
All three agree with $\sigma_{\den{c_3}} = \{a \mapsto 2,\,
c \mapsto 9\}$, as Theorem~\ref{thm:cut} requires.
\end{example}

\subsection{Per-key validity is strictly weaker than a single read
point}\label{sec:theorem-torn}

Section~\ref{sec:contract-obligations} called O2 the central weakening. The
following instance shows that the weakening is strict. Its refresh
combines two committed observations made at different coordinates. No
single coordinate accounts for both reads, yet the contract admits
the instance and its cut still holds.

\begin{proposition}[No shared read point]\label{prop:torn}
There is a contract instance with one unit over two keys such that
\begin{enumerate}
\item no single coordinate $c$ in the unit's bracket satisfies
$R(k) = \sigma_{\den{c}}(k)$ for both keys $k$, yet
\item the instance satisfies O1-O5, and
\item its canonical replay equals the row-event replay state at the frontier.
\end{enumerate}
\end{proposition}

\begin{proof}
Over an initially empty store, take two keys $a, b$ and the
two-event history $L = e_1\, e_2$ with $\keyof(e_1) = a$,
$\imgof(e_1) = 1$, $\keyof(e_2) = b$, $\imgof(e_2) = 2$, and
coordinates $c_0, c_1, c_2$ mapping to the prefixes of lengths $0, 1,
2$. The plan has one unit: domain $\{a, b\}$, bracket $(c_0, c_2)$,
frontier $f = c_2$, and refresh
\[
R(a) = \perp, \qquad R(b) = 2.
\]
The read of $a$ happened before $e_1$ committed and the read of $b$
after $e_2$ committed. Both observations are of committed state, but
they reflect different source coordinates.

\emph{Contract.} O1 holds with the single unit. O3 holds since $R$
is total. O5 holds by construction. For O2, key $a$ matches the 
state at $c_0$, where $\sigma_{\den{c_0}}(a) = \perp = R(a)$, and key $b$ 
matches the state at $c_2$, where $\sigma_{\den{c_2}}(b) = 2 = R(b)$. Both
witnesses lie in the bracket $(c_0, c_2)$.

\emph{No single shared read coordinate.} The bracket contains exactly the
coordinates $c_0, c_1, c_2$. At $c_0$ and $c_1$ the source has
$\sigma(b) = \perp \neq 2 = R(b)$. At $c_2$ it has
$\sigma(a) = 1 \neq \perp = R(a)$. No in-bracket coordinate is
consistent with both keys, so an obligation demanding one read point per
unit rejects this plan.

\emph{Cut.} The window $\window{c_0}{c_2}$ contains $e_1$ and
$e_2$, so both keys fall under Case~A of Theorem~\ref{thm:cut},
with $\sink_f(a) = \imgof(e_1) = 1$ and $\sink_f(b) = \imgof(e_2) = 2$,
which is exactly $\sigma_{\den{c_2}}$. The refresh results do not
decide the final fold.
\end{proof}

The proposition also shows why the classic algorithm's
single-coordinate chunk reads satisfy this contract directly. A
single read point implies O2 by
instantiating every key's witness uniformly, while O2 does not imply
a single read point. The weakening is strict, and it is the exact
degree of freedom needed when an implementation establishes per-key
bracket condition O2 without establishing a single shared read coordinate.
The distinction is between proof obligations. Ordinary chunk selects
may well run under statement snapshots. The read-only and parallel
instances of Section~\ref{sec:instances} are placed without assuming one,
since the isolation of their chunk selects is undocumented
(\S\ref{sec:fence}).

\subsection{Corollary 1: a shared witness can move the canonical
trajectory earlier}\label{sec:theorem-shared}

The frontier-parametric proof constrains more than one state. To make
that fact and the additional role of a shared witness precise, we
first extend the canonical replay to intermediate coordinates.

\begin{definition}[Canonical replay at a coordinate]\label{def:traj}
For any coordinate $g$ and key $k \in D_i$, the \emph{canonical
replay at} $g$ is
\[
\sink_g(k) \;=\;
\begin{cases}
\imgof(e^{\ast}) & \parbox[t]{0.5\columnwidth}{if $e^{\ast}$, the
  last event for $k$ in $\window{lo_i}{g}$, exists,}\\[2pt]
R_i(k) & \text{otherwise,}
\end{cases}
\]
where $\window{lo_i}{g}$ is read as $\den{g} \setminus \den{lo_i}$
for an arbitrary pair of coordinates. When $g \pre lo_i$ it is
empty and the refresh stands. We call the family $g \mapsto \sink_g$
the plan's \emph{canonical trajectory}.
\end{definition}

At $g = f$ this is the canonical replay of
\S\ref{sec:contract-replay}. For general $g$ it is a retrospective
family derived from the final plan. Every unit contributes its
refresh results, and events up to $g$ supersede them. It is not, by
definition, the sequence of states produced while an emitted stream
is consumed. Calling those physical sink prefixes a trajectory would
require prefix-equivalence at every $g$, or an atomic or ordered
handoff condition, beyond the final equivalence obligation of
\S\ref{sec:contract-replay}.

\begin{lemma}[Latest-high trajectory]\label{lem:latest-high}
Let a nonempty finite capture plan at frontier $f$ satisfy the
contract. Among its recorded high edges choose one $h$ whose corresponding
prefix is maximal. Then $h \pre f$, every unit satisfies
$hi_i \pre h$, and for every $g$ with $h \pre g \pre f$ and every
$k \in \keys$,
\[
\sink_g(k) \;=\; \sigma_{\den{g}}(k).
\]
Thus the canonical trajectory of every nonempty T2 plan matches the
source from its latest high edge onward.
\end{lemma}

\begin{proof}
The unit family is finite and coordinates map to linearly
nested prefixes, so the nonempty set of prefixes at its high edges has a
maximum. Choose a recorded high edge $h$ that maps to that prefix. Every
$hi_i \pre h$, and $h \pre f$ because every unit high is bounded by
the frontier. For any $g$ with $h \pre g \pre f$, transitivity gives
$hi_i \pre g$ for every unit. The frontier-parametric form of
Theorem~\ref{thm:cut} therefore applies at $g$.
\end{proof}

If the unit family is empty, O1 forces the capture scope to be empty.
The scoped equality is then vacuous, but there is no recorded high
edge to select. The rest of the paper concerns nonempty captures.

The proof needs the high-edge bound only for the unit that owns the
current key. This yields the mid-capture form illustrated in
Figure~\ref{fig:cuts}.

\begin{lemma}[Closed-unit cut]\label{lem:closed}
Let a source, a coordinate space, and a capture plan at frontier $f$
satisfy the contract, let $g$ be any coordinate, and let $k \in D_i$
be a key whose unit has closed by $g$, that is, $hi_i \pre g$. Then
$\sink_g(k) = \sigma_{\den{g}}(k)$. Consequently, at every coordinate
$g$, the canonical replay agrees with the row-event replay state on
every key owned by a unit whose high edge lies at or before $g$,
whether or not the other units have closed.
\end{lemma}

\begin{proof}
The proof of Theorem~\ref{thm:cut} fixes $k$ and its covering unit $i$
and uses the frontier only through $hi_i \pre f$. With $g$ in place
of $f$ and Definition~\ref{def:traj} supplying $\sink_g$, Case~A and
Case~B apply verbatim. The last event for $k$ in $\window{lo_i}{g}$
decides both sides, or, if there is none, O2 and Lemma~\ref{lem:inv}
carry the valid state from the bracket to $g$ because
$c \pre hi_i \pre g$. No fact about any other unit is used.
\end{proof}

Figure~\ref{fig:cuts} illustrates the lemma. Once a unit closes, its
keys are exact at the close and remain exact as later events are
replayed, while keys in open units carry no claim. Under their respective
observation and closure conditions, the same per-key argument applies to
the corresponding prefixes of the window-discard and range-merge
emissions (Lemmas~\ref{lem:discard-shape} and~\ref{lem:rm-shape}).

\begin{corollary}[Shared witness]\label{cor:shared}
Let a source, a coordinate space, and a capture plan at frontier
$f$ satisfy the contract, and suppose some coordinate $c^{\ast}$
is a \emph{shared state witness} for the whole plan:
\begin{itemize}
\item $lo_i \pre c^{\ast} \pre hi_i$ for every unit $i$, and
\item $R_i(k) = \sigma_{\den{c^{\ast}}}(k)$ for every unit $i$ and
every $k \in D_i$, meaning that each refresh matches the source state \emph{at} $c^{\ast}$, not
merely somewhere in its own bracket.
\end{itemize}
Then for every $g$ with $c^{\ast} \pre g \pre f$ and every
$k \in \keys$,
\[
\sink_g(k) \;=\; \sigma_{\den{g}}(k):
\]
the canonical replay trajectory coincides, key by key, with the
source's row-event replay trajectory from $c^{\ast}$ up to the
observation frontier.
\end{corollary}

\begin{proof}
Fix $g$ with $c^{\ast} \pre g \pre f$, fix $k$, and let $i$ be the
covering unit. Note $lo_i \pre c^{\ast} \pre g$. If some event for
$k$ lies in $\window{lo_i}{g}$, the last such event decides both
sides exactly as in Case~A of Theorem~\ref{thm:cut}, with $g$ in
place of $f$. Otherwise the window $\window{lo_i}{g}$ contains no event for $k$. Since $\window{c^{\ast}}{g} \subseteq \window{lo_i}{g}$, the window $\window{c^{\ast}}{g}$ contains no event for $k$ either, and
Lemma~\ref{lem:inv} gives
$\sigma_{\den{g}}(k) = \sigma_{\den{c^{\ast}}}(k)$. Hence
\[
\sink_g(k) \;=\; R_i(k)
          \;=\; \sigma_{\den{c^{\ast}}}(k)
          \;=\; \sigma_{\den{g}}(k),
\]
using the shared state witness for the middle equality.

\end{proof}

The additional information is the \emph{exhibited common state
witness}. Lemma~\ref{lem:latest-high} already gives every nonempty T2
plan a conservative trajectory from its latest high $h$.
Corollary~\ref{cor:shared} identifies one coordinate $c^{\ast}$ at
which every refresh agrees with the same source state. Because
$c^{\ast} \pre hi_i \pre h$ for every unit, it can move the guaranteed
start earlier than $h$. The start is not necessarily earlier, since
$c^{\ast}$ may equal $h$. We call it the plan's \emph{exhibited
shared-state anchor}, and T3 records this additional evidence
(\S\ref{sec:tiers}).

The formal plan does not place it before the physical reads or claim
that physical sink prefixes follow this history. Transactional-history
language additionally requires transaction closure throughout the
claimed range. Beyond $f$ the canonical trajectory extends as
far as observation is extended (Corollary~\ref{cor:cont}).

Operationally, a shared state witness arises in three ways. The first is a
shared point bracket, where $lo_i = hi_i = c^{\ast}$ for every unit
and the windows are empty, so there is nothing to merge and the splice
determines everything. The second is a shared engine snapshot read
inside wider brackets, as when every unit's select runs against one
exported or MVCC read view whose transactions form a prefix of the log.
The third is a quiesced window, where
writes are simply held off between the reads. The splice instances of
\S\ref{sec:instances} are the first case as it occurs in production.

\subsection{Corollary 2: continuation under extended
observation}\label{sec:theorem-cont}

\begin{corollary}[Continuation]\label{cor:cont}
Let a source, a coordinate space, and a capture plan at frontier
$f$ satisfy the contract, and let $f \pre f'$. Then
$\sink_{f'}(k) = \sigma_{\den{f'}}(k)$ for every $k \in \keys$.
\end{corollary}

\begin{proof}
Every unit satisfies $hi_i \pre f \pre f'$, so $f'$ bounds every
bracket's high edge, and the frontier-parametric form of
Theorem~\ref{thm:cut}'s proof applies verbatim with $f'$ in place
of $f$.
\end{proof}

The extension costs observation. The
corollary's fold is a statement about the committed history. A
capture \emph{realizes} it as far as it keeps observing,
which is S-OBS extended to $f'$. The fold beyond the frontier
consumes events a stopped capture no longer sees, and no algebra
recovers them. Steady-state operation, where a bootstrap
completes and the pipeline then follows the log indefinitely, is
Corollary~\ref{cor:cont} applied at ever-larger $f'$. The
guarantee survives as long as consumption and retention
do.

\subsection{Corollary 3: mixing and composition guarantees}
\label{sec:theorem-mix}

Nothing in Definition~\ref{def:plan} records how a refresh was
obtained, so a single plan may already mix a chunked scan here, a
dump there, and a point select elsewhere
(Figure~\ref{fig:mixing}).
Theorem~\ref{thm:cut} quantifies over all of them at once.
Heterogeneous bootstrap is therefore not a special case requiring
its own theory. It is the theorem's normal case, and the only
composition fact left to check is that conforming plans over
disjoint scopes can be pasted into one conforming plan.

\begin{figure}[t]
\centering
\includegraphics{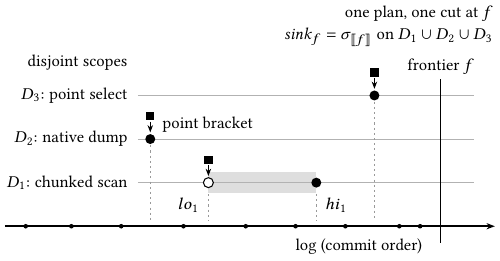}
\caption{A heterogeneous plan at one frontier. Three units over
disjoint scopes mix provenances, including a chunked scan with an ordinary
bracket, alongside a native dump and a point select each behind a point
bracket of its own. No coordinate lies inside all three brackets,
so no shared state witness is available. The composite nevertheless
satisfies the contract. One cut at $f$ covers the union of the
scopes, and its canonical trajectory is guaranteed from the
latest of the three high edges (Corollary~\ref{cor:mix},
Theorem~\ref{thm:cut}, and Lemma~\ref{lem:latest-high}).}
\label{fig:mixing}
\Description{Three horizontal unit rows above one commit-order log
axis. The chunked-scan row carries an ordinary bracket with a
shaded window. The native-dump and point-select rows each carry a
point bracket at staggered positions outside the scan's bracket. A
single vertical frontier line crosses all rows, and one claim at
the frontier states the cut over the union of the three scopes.}
\end{figure}

\begin{corollary}[Mixing]\label{cor:mix}
Let two capture plans over the same source and coordinate space
satisfy the contract at the same frontier $f$, on disjoint scopes
$\keys_1$ and $\keys_2$. Then their union, with units combined and
scope $\keys_1 \cup \keys_2$, satisfies the contract at $f$, and
consequently replays to $\sigma_{\den{f}}$ on all of
$\keys_1 \cup \keys_2$.
\end{corollary}

\begin{proof}
Each obligation decomposes over the union. For coverage, the combined
domains cover $\keys_1 \cup \keys_2$ because each family covers its
own scope. For disjointness, two units from the same plan have disjoint
domains by that plan's O1. Two units from different plans have
disjoint domains because each domain is contained in its plan's
scope and the scopes are disjoint. Brackets and their frontier
bounds are per-unit facts and carry over unchanged, as do O2's
per-key witnesses, O3's totality, and O5's typing. O4's declared
assumptions combine by-union, and S-LOG and S-IMG are
facts of the one shared source and log, common to both plans. The
union therefore satisfies
Definition~\ref{def:contract}, and Theorem~\ref{thm:cut} applied to
it gives the cut on the combined scope.
\end{proof}

The same argument gives one more decomposition. If a single
coordinate $c^{\ast}$ is a shared state witness for \emph{every}
member of both plans, then it is a shared state witness for the 
union, since the two hypotheses of Corollary~\ref{cor:shared} are per-unit 
statements. The composite then carries the full canonical trajectory property 
from $c^{\ast}$. Sharing one engine snapshot across all members,
its coordinate recorded, lifts the composite to T3.

Without a witness shared by every member, composition alone supplies
the T2 cut. When the combined unit family is nonempty, it also supplies
the conservative trajectory from the union's latest high. It does not automatically supply one shared-state anchor
or an earlier start. A shared witness may still exist for a particular
plan, but it needs separate evidence. The sharpness obstruction has
the exact shape of per-table initial
loads taken at distinct coordinates. Each table's copy matches that
table's row-event replay state at its own coordinate, but the union
may fail to match the combined scope at any one coordinate. We call this the
\emph{tablesync} shape. Section~\ref{sec:tiers} categorizes the real systems
that have it.

\begin{proposition}[Composition can lose every earlier onset]\label{prop:sync}
There is a contract instance composed of two units, each of which,
as a one-unit plan on its own scope, satisfies 
the hypotheses of Corollary~\ref{cor:shared} at its own witness, while the
composite has no shared in-bracket coordinate and admits \emph{no}
early onset. For every coordinate
$c^{\ast}$ strictly below the frontier, the trajectory property of
Corollary~\ref{cor:shared} fails at some $g$ with
$c^{\ast} \pre g \pre f$.
\end{proposition}

\begin{proof}
Over an initially empty store, take the history $L = e_1\, e_2$
with $\keyof(e_1) = a$, $\imgof(e_1) = 1$, $\keyof(e_2) = b$,
$\imgof(e_2) = 2$, coordinates $c_0, c_1, c_2$ as in
Proposition~\ref{prop:torn}, and frontier $f = c_2$. Compose two
point-bracket units:
\[
\begin{aligned}
A &: \text{domain } \{a\},\; lo_A = hi_A = c_0,\; R_A(a) = \perp,\\
B &: \text{domain } \{b\},\; lo_B = hi_B = c_2,\; R_B(b) = 2.
\end{aligned}
\]
Unit $A$ reads before anything commits and records $a$
absent. Unit $B$ reads at the end of history and records
$b = 2$. Each read is valid at its unit's point coordinate, so
each one-unit plan satisfies the hypotheses of Corollary~\ref{cor:shared} 
with witness $c_0$ respectively $c_2$, and the composite satisfies
the contract by Corollary~\ref{cor:mix}. Theorem~\ref{thm:cut}
gives its cut at $c_2$.

The point brackets $[c_0,c_0]$ and $[c_2,c_2]$ have no coordinate in
common, so the composite has no shared in-bracket witness. Its latest
high is $h=c_2=f$, and Lemma~\ref{lem:latest-high} gives the
conservative, degenerate trajectory interval $[f,f]$.

Now evaluate the composite's trajectory at the intermediate
coordinate $g = c_1$. Unit $B$'s window $\window{c_2}{c_1}$ is
empty, so $\sink_{c_1}(b) = R_B(b) = 2$, although
$\sigma_{\den{c_1}}(b) = \perp$. The canonical replay asserts a
value for $b$ absent from the source's row-event replay at $c_1$. The
trajectory property therefore fails at $g = c_1$, and since any
onset $c^{\ast}$ strictly below the frontier satisfies
$c^{\ast} \pre c_1 \pre c_2 = f$, no such onset exists.
\end{proof}

The resulting composition rule is short. Every conforming composite
reaches T2 by Corollary~\ref{cor:mix} and Theorem~\ref{thm:cut}. For a
nonempty combined unit family, Lemma~\ref{lem:latest-high} also supplies
its conservative trajectory.
Proposition~\ref{prop:sync} shows that composition alone cannot promise
a shared witness or any earlier onset. One exhibited witness shared by
all members establishes T3 (Corollary~\ref{cor:shared} via the
decomposition above). Bracket
\emph{sharing}, by contrast, is unconstrained by the contract.
Definition~\ref{def:plan} already allows distinct units to carry
identical brackets, no obligation counts brackets, and one bracket
amortized over many units with one shared state witness is exactly 
the hypothesis of Corollary~\ref{cor:shared}. Amortizing a single bracket 
over many reads is therefore level-preserving, which is what makes
batching brackets a pure cost optimization
(\S\ref{sec:instances}).

\section{Guarantee Levels and Certification}\label{sec:tiers}

Levels T0 through T3 describe what has been established about a
configuration's replay and witness structure. T4 records a different
dimension, whether a T2 or T3 claim carries an externally checkable
certificate. It is therefore a certification overlay rather than a
stronger replay semantics. Each production shape in
\S\ref{sec:instances} receives the applicable semantic level and, when
available, the certification overlay. Mechanisms not placed here remain
explicitly open.

These levels separate what can be \emph{proved} from how a configuration
performs. Most knobs operators turn do not change either dimension.
Figure~\ref{fig:ladder} draws the semantic progression, the independent
certification path, and the cost knobs that belong to neither.

Operationally, T2 means that the plan's canonical replay is correct at
its frontier and, after the latest recorded high edge, along its
canonical trajectory. An implementation's emitted stream
achieves that result only after its equivalence proof is established.
T3 additionally exhibits one coordinate at which all copied units
agree with the same source state, which may move the trajectory start
earlier. T0 and T1 establish no source-side cut.

\subsection{Four semantic levels and one certification overlay}\label{sec:tiers-rungs}

\begin{itemize}
\item \textbf{T0: outside the contract.} The configuration is not an
instance. Either no faithful prefix reading of its
coordinates exists (the admissibility clause of
\S\ref{sec:setting-coords}, where the system's own membership decisions
agree with no mapping to commit-order prefixes), or S-LOG, S-IMG, or a plan
obligation provably fails with no declared repair. Examples include deletes invisible
to the consumed stream or an enumeration provably incomplete. Note
what T0 is not. A lawful plan with a wide bracket, or even one
run-global bracket, is still an instance. Width never moves a plan to T0, and a global bracket is still a bracket.
\item \textbf{T1: cut not established.} The configuration's
guarantee argument is incomplete as shipped. Bracket-local
validity, domain completeness, or the merge discipline's equivalence
obligation (\S\ref{sec:merge}), its S-OBS floor included, is not
established. The framework
then claims nothing beyond the emission's content, establishing neither a cut nor
per-key monotonicity. Such deployments may nevertheless converge by
their own unproved mechanism or through machinery external to this
contract, such as version-stamped apply and checksum repair. That
machinery is sink-side and stays outside the modeled scope
(\S\ref{sec:fence}). Dump-then-verify pipelines can live here, as do
backfill-skipping modes whose own documentation promises
at-least-once delivery only (\S\ref{sec:instances}). The level judges
the argument, not the software. A buffering service whose equivalence
obligation is undischarged sits at T1, and it moves to T2 as soon as
the obligation is satisfied.
\item \textbf{T2: per-key replay equivalence at the frontier.}
At Theorem~\ref{thm:cut}'s level, canonical replay equals the source's
row-event replay state at the frontier, key by key, on the whole
scope. Every contract
instance reaches T2 as a plan, whatever its provenance,
parallelism, or
bracket geometry. The cut theorem requires nothing else. A
shipped configuration stands here once its plan satisfies the
contract and its discipline's equivalence obligation is
established (the second choice of \S\ref{sec:tiers-choices}). With
the equivalence undischarged, the same deployment is the T1 case
above. Reading this row-event cut as a transactional snapshot
additionally requires a transaction-closed frontier. For every
nonempty plan, Lemma~\ref{lem:latest-high} also guarantees the
canonical trajectory from the latest recorded high edge through the
frontier.
\item \textbf{T3: exhibited shared-state anchor.}
This is Corollary~\ref{cor:shared}'s level, claimed at a named coordinate.
Some $c^{\ast}$ lying inside every unit's bracket is a shared state
witness, \emph{exhibited by the capture} through a recorded binding
between the copied state and that coordinate. From $c^{\ast}$ onward,
the plan's canonical trajectory
coincides with the source's row-event replay. Because
$c^{\ast} \pre hi_i$ for every unit, this can begin earlier than T2's
conservative latest-high start. Equality is possible, so the semantic
interval is not strictly longer in every instance. The additional
guarantee is the shared-state evidence itself.
This level does not assert that physical reads were simultaneous or
that transient prefixes of an emitted stream follow that trajectory.
The latter requires prefix-equivalence or an atomic, ordered handoff.
A transactional-history reading additionally requires transaction
closure at the coordinates considered.
\item \textbf{T4: certification overlay.} T2 or T3 together with an
externally checkable certificate whose acceptance supports
the cut, provided the checker assumptions hold and the required source 
history and chunk reads are recorded correctly.
Certification neither creates a shared witness nor moves the trajectory start. 
The certified virtual-cut object
of the classic algorithm's published
formalization~\cite{andreakis2026virtualcuts} anchors this overlay and
is the formalized example used here. Section~\ref{sec:instances} anchors the T4
overlay there.
\end{itemize}

\begin{figure}[H]
\centering
\includegraphics{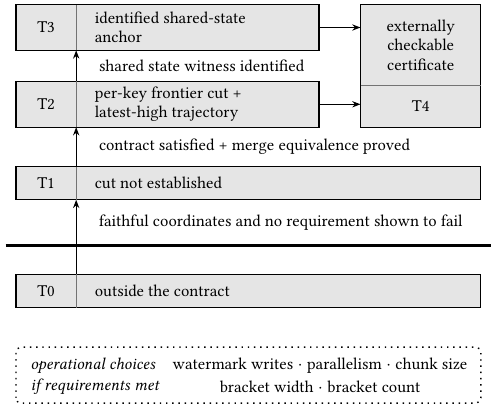}
\caption{Guarantee structure. Three choices determine the semantic
level. A
faithful coordinate reading, with no disproved source promise or plan
obligation, crosses the T0 exclusion boundary. Fulfilling all contract 
obligations together with a proven merge equivalence guarantees the cut for the emitted stream. An
identified shared state witness adds one shared-state anchor and
may move the conservative start earlier. Certification applies
independently to either T2 or T3. Holding those conditions fixed, the
operational choices change neither semantic level nor certification.}
\label{fig:ladder}
\Description{Four stacked semantic bars labeled T0 through T3, with
T2 identifying the per-key frontier cut and post-latest-high trajectory and
T3 adding an identified shared-state anchor. A separate certification
box receives arrows from both T2 and T3. A dotted, unconnected box
lists the operational choices.}
\end{figure}

\subsection{What changes a level, and what does
not}\label{sec:tiers-choices}

Three choices determine the semantic level:

\begin{enumerate}
\item a faithful coordinate reading, with neither S-LOG, S-IMG,
nor any plan obligation disproved without a declared repair (the
T0 exclusion boundary).
\item all contract obligations, including bracket-local validity,
domain completeness, and image sufficiency, plus a proven
merge-discipline equivalence (T1 versus T2).
\item an exhibited shared state witness (T2 versus T3).
\end{enumerate}

Certification is a separate yes-or-no dimension. An externally
checkable certificate adds the T4 overlay to either T2 or T3 without
changing the underlying replay or witness claim.

Holding the contract, S-OBS, and emission-equivalence conditions fixed,
write-freeness, parallelism, chunk size, bracket width, and bracket
count affect cost rather than theorem strength. Provenance and
scheduling appear nowhere in the plan signature, and no theorem 
condition depends on the width of the bracket pair. A wider 
bracket can, however, enlarge reconciliation work and the log span
that S-OBS must keep readable. Cost choices can therefore determine
whether a deployment can actually satisfy the required conditions 
and establish a level, even though the theorem does not rank them 
directly.

\subsection{Reading the levels on a mixed bootstrap}
\label{sec:tiers-example}

\begin{example}[Composition across levels]\label{ex:tiers}
An operator bootstraps four tables. Tables $t_1, t_2, t_3$ are
loaded from per-table dumps, each a point bracket at its own
recorded coordinate. Table $t_4$ is chunk-scanned with ordinary
brackets. Each dump alone satisfies the hypotheses of Corollary~\ref{cor:shared} 
on its own scope at its own 
coordinate, so each is T3 on its own table. The composite of all
four is a contract instance by Corollary~\ref{cor:mix} and reaches
T2 on the union, with one virtual cut at the shared frontier and a
conservative trajectory from the latest recorded high. The member
witnesses do not automatically yield T3 on the union. The three dump
coordinates are distinct, and Proposition~\ref{prop:sync} exhibits
this shape with no shared witness and no onset before the frontier.
This is the tablesync shape, named
after the per-table initial-copy workers of logical
replication~\cite{postgres_logical_replication}. Every table matches
its own coordinate, but the union may fail to match the combined scope
before the frontier. One change lifts it. Running the three dumps against a single
shared engine snapshot whose coordinate is recorded makes that
exhibited coordinate a witness for
all three units, so the dumped scope regains T3 while $t_4$'s chunks
still mix in at T2 for the union. The choice that changed the level
was the witness structure and nothing else. Dump versus scan,
serial versus parallel, and bracket width do not alter the theorem
once their observation and equivalence conditions hold.
\end{example}

\subsection{T0 in the wild}\label{sec:tiers-t0}

The polling class shows the boundary is real. A connector that
periodically issues
$\texttt{SELECT} \ldots \texttt{WHERE updated\_at} > t$ and emits
the returned rows consumes no commit-ordered log, because its only
signal is each row's own last-modified stamp. Deletes are invisible,
since a vanished row simply stops appearing and no event ever
reports its absence. S-LOG has nothing to promise over. Nor do the
poll timestamps admit a faithful prefix reading. The stamps are
assigned by clients or by statement time rather than at commit, so
rows become visible out of stamp order and the connector's own
membership decision, ``stamp greater than my last poll,''
disagrees with every mapping to commit-order prefixes. Both failures are
independent, either alone is disqualifying, and neither is
repairable by widening any bracket. The class is T0, lying outside the family rather than being a weak member of it.

\section{Merge Semantics}\label{sec:merge}

The canonical rule of \S\ref{sec:contract-replay} defines the
meaning of an emission. An implementation may emit a different
stream, as both buffering disciplines do. To use the cut theorem,
the implementation must show that replaying its stream yields the
same $\sink_f$ as the canonical rule. This section proves that
equivalence for the disciplines used by the five instances. It
also states, without claiming proofs, the obligations for four
further disciplines. The withholding rule of
\S\ref{sec:merge-range} is one correctness condition that becomes
visible only at this layer. Throughout the section we fix a source,
a coordinate space, and a conforming capture plan at frontier $f$.
Each discipline adds its own consumption start and emission clauses.

\subsection{Streams and their replay}\label{sec:merge-streams}

\begin{definition}[Stream replay]\label{def:stream}
An \emph{emitted stream} is a finite sequence of entries, each
carrying a key and a value-or-absence result. Log events serve as
entries directly. Its \emph{replay} is the per-key fold from the
empty state. A key's last entry decides its value, and a key with
no entry replays to absent.
\end{definition}

\begin{corollary}[Actual emitted replay at the frontier]
\label{cor:emitted-cut}
Let a source, coordinate space, and capture plan at frontier $f$
satisfy the contract. Let $E$ be a finite emitted stream. If
\[
\mathrm{replay}(E)(k)=\sink_f(k)
\qquad\text{for every }k\in\keys,
\]
then
\[
\mathrm{replay}(E)(k)=\sigma_{\den{f}}(k)
\qquad\text{for every }k\in\keys.
\]
\end{corollary}

\begin{proof}
For each scoped key, substitute the assumed emission equivalence into
Theorem~\ref{thm:cut}.
\end{proof}

Corollary~\ref{cor:emitted-cut} is the reusable bridge from a modeled
or actual finite output to the source-side result. It states final
replay from an empty state. It does not identify physical output
prefixes with Definition~\ref{def:traj}, prove transport or delivery,
or explain application to a populated sink. Reading its conclusion as
a transactional snapshot additionally requires the frontier to be
transaction-closed.

Starting from the empty state is what lets absence by omission mean absence.
A physical scan emits only the finitely many present
rows it returns. Its complete range membership induces implicit
absence for every other identity in the unit, even when the logical
domain itself is not finite. A discipline therefore never emits
refresh entries for keys read as absent, whereas delete events are
emitted because they arrive from the log. Both routes represent the same
state, and the equivalence proofs below rely on that identification.
A repair applied to a populated sink cannot use omission this way. It
must clear the repaired scope before replay or materialize source
absence as explicit tombstones. Sink-side execution remains outside
the theorem.

One finiteness convention accompanies it. Each unit
has finitely many rows to emit at its close, as every real
chunk, a finite set of physical rows, does. Close blocks are
therefore finite, and every emission below is the finite sequence
Definition~\ref{def:stream} requires. Beyond its finite family of
units, the contract layer imposes no finiteness requirement on unit
domains or refresh maps. That additional condition is needed only at
the stream layer.

\subsection{Window-discard}\label{sec:merge-discard}

The classic discipline buffers a unit's refresh and applies the rule that the log wins. On any in-window event for a buffered
key, the buffer entry is discarded and the event flows through. When
the window closes, the still-buffered rows are emitted. The keys that survive to the close are exactly those with no event in the window.

\begin{definition}[Survivors and window-discard emission]
\label{def:discard}
For a unit $i$, the \emph{survivors} are the keys
$k \in D_i$ with $R_i(k) \neq \perp$ and no event in
$\window{lo_i}{hi_i}$. Fix a consumption start $s_0$ with
$s_0 \pre lo_i$ for every unit. The \emph{window-discard emission}
is one pass over the consumed events $\window{s_0}{f}$ in commit
order, emitting every event at its own place. Immediately before
the first consumed event beyond a unit's high edge, or at the end
of the pass if none exists, the unit \emph{closes}. One refresh
entry per survivor is emitted, carrying the read result.
Distinct units never share a survivor key, by O1, so the order of
same-position close blocks is immaterial.
\end{definition}

\begin{lemma}[Per-key shape of the discard emission]
\label{lem:discard-shape}
Let $k \in D_i$. The entries for $k$ in the window-discard
emission are: $k$'s consumed events in $\window{s_0}{hi_i}$, in
commit order. Then, exactly if $k$ is a survivor of unit $i$, one
refresh entry with result $R_i(k)$. Then $k$'s events in
$\window{hi_i}{f}$, in commit order.
\end{lemma}

\begin{proof}
Events for $k$ pass through at their own places, and the pass
visits them in commit order. Only unit $i$ can emit a refresh
entry for $k$, because a survivor entry for $k$ presupposes
$k \in D_j$ for the emitting unit $j$, and O1 gives $j = i$. That
entry, when it exists, is emitted at unit $i$'s close, which sits
after every consumed event at or before $hi_i$ and before every
consumed event beyond $hi_i$.
\end{proof}

\begin{theorem}[Window-discard equivalence]\label{thm:discard}
For every scoped key, the replay of the window-discard emission
equals the canonical replay:
$\mathrm{replay}(\text{emission})(k) = \sink_f(k)$ for all
$k \in \keys$.
\end{theorem}

\begin{proof}
Fix $k \in D_i$ and write $W = \window{lo_i}{f}$ for the canonical
window. Lemma~\ref{lem:discard-shape} lists $k$'s stream entries.
We compare last entries against the canonical fold, splitting on
where $k$'s events fall.

\emph{Some event for $k$ lies in $\window{hi_i}{f}$.} The last
such event is $k$'s last stream entry, since post-close events
follow everything else. It is also the last event of $W$, because
$W$ consists of $\window{lo_i}{hi_i}$ followed by
$\window{hi_i}{f}$. Both replays return its image.

\emph{No post-close event, but some event for $k$ lies in
$\window{lo_i}{hi_i}$.} Then $k$ is not a survivor, so its stream
entries are exactly its consumed events in $\window{s_0}{hi_i}$,
and the last of these is the last event of
$\window{lo_i}{hi_i}$, which in turn is the last event of $W$.
Again both replays return its image.

\emph{No event for $k$ in $W$ at all.} The canonical side returns
$\sink_f(k) = R_i(k)$, and by O2, validity at some in-bracket $c$
together with Lemma~\ref{lem:inv} gives
$R_i(k) = \sigma_{\den{lo_i}}(k)$. Two sub-cases on the read result.
If $R_i(k) \neq \perp$, then $k$ is a survivor. Its refresh entry
is its last stream entry and the stream replays to $R_i(k)$
directly. If $R_i(k) = \perp$, then $k$ is not a survivor and its
stream entries are its consumed events in $\window{s_0}{lo_i}$, of
which there may be none. With none, the stream replays $k$ to
absent, which equals $R_i(k)$. Otherwise let $q$ be the last such
event. $q$ is then the last event for $k$ in $\den{lo_i}$, so
$\imgof(q) = \sigma_{\den{lo_i}}(k) = R_i(k) = \perp$. The entry
is a delete, and the stream replays $k$ to absent once more. In
every sub-case the stream replay equals $R_i(k) = \sink_f(k)$.
\end{proof}

The stream may carry pre-bracket events for a key that the canonical
rule never consults, because consumption started before the bracket did. 
Condition O2 ensures that those events are consistent with the read result.
A key read as absent can reach the sink either by never being 
mentioned or by ending on its own delete, and the two routes agree. 
Symmetrically, a surviving refresh entry
is emitted \emph{after} any pre-bracket event for its key and
overrides it, which is the ordering
Definition~\ref{def:discard} builds in by closing units inside the
pass rather than up front.

\begin{theorem}[Window-discard monotonicity]\label{thm:discard-mono}
Tag each emitted entry with a count of committed events, tagging each event
with the length of the log prefix ending at it and each surviving
refresh entry with the length of $\den{hi_i}$, its unit's close.
Then along the window-discard emission,
each key's tags are non-decreasing, so no entry ever replays a key to
an older state after a newer one.
\end{theorem}

\begin{proof}
By Lemma~\ref{lem:discard-shape}, $k$'s entries are its pre-close
events, whose tags are at most the length of $\den{hi_i}$ and
appear in
commit order. Then possibly the refresh entry, tagged exactly that
length. Then
its post-close events, in commit order, with tags beyond it.
Each block is internally non-decreasing and each boundary
steps upward.
\end{proof}

The surviving refresh's tag is exact by
Lemma~\ref{lem:inv}. A survivor's window contains no event for the key, so a read result valid somewhere in the bracket matches the key's state at the high edge.

\subsection{Buffered range-merge}\label{sec:merge-range}

The second concrete merge discipline reads a unit's window as a ranged
backfill, upserts the window's events onto the refresh buffer, and
emits the merged unit once, as plain inserts. The live stream
carries only events beyond the unit's high edge. Two clauses make the discipline lawful, and both are necessary.

\begin{definition}[Merged result]\label{def:merged}
For a unit $i$ and $k \in D_i$, the \emph{merged result} is
\[
M_i(k) \;=\;
\begin{cases}
\imgof(e^{\ast}) & \parbox[t]{0.47\columnwidth}{if $e^{\ast}$, the
  last event for $k$ in $\window{lo_i}{hi_i}$, exists,}\\[2pt]
R_i(k) & \text{otherwise:}
\end{cases}
\]
the buffer after the upsert, where the key's last in-window event decides and a key with no in-window event keeps its read result. A
delete-then-reinsert key is present again, because only the last
event counts.
\end{definition}

\begin{lemma}[The merged buffer is the state at the high edge]
\label{lem:merged}
For every unit $i$ and every $k \in D_i$,
$M_i(k) = \sigma_{\den{hi_i}}(k)$.
\end{lemma}

\begin{proof}
If some event for $k$ lies in $\window{lo_i}{hi_i}$, the last one
is also the last event for $k$ in $\den{hi_i}$, by the split of
$\den{hi_i}$ at $lo_i$, and both sides are its image. Otherwise
$M_i(k) = R_i(k)$, and O2 with Lemma~\ref{lem:inv} carries the
in-bracket read value to the high edge:
$R_i(k) = \sigma_{\den{hi_i}}(k)$.
\end{proof}

This makes the discipline's clause (b) exact. The merged unit
represents each key's value or absence at the high edge. Keys absent at
the high edge are emitted zero times, which under
Definition~\ref{def:stream} replays as absence.

\begin{definition}[Gated range-merge emission]\label{def:rm}
Fix a consumption start $s_0$ with $s_0 \pre hi_i$ for every unit.
The \emph{range-merge emission} is one pass over the consumed
events $\window{s_0}{f}$ in commit order, subject to the
\emph{withholding gate}, clause (a), under which an event whose key belongs to
unit $i$ passes to the stream only beyond unit $i$'s high
edge, outside $\den{hi_i}$. Events at or before the high edge for the 
unit's keys are withheld and never emitted individually. At unit $i$'s close,
placed as in Definition~\ref{def:discard}, one entry per key with
$M_i(k) \neq \perp$ is emitted, carrying $M_i(k)$. The merged
block is the unit's sole emission at or before its high edge.
\end{definition}

The observability requirement is genuinely weaker here than for
window-discard, and the difference matters for the parallel instance. Window-discard consumes one pass from
$s_0 \pre lo_i$. Range-merge reads each window by a coordinate-
addressable backfill (the S-OBS retention clause of
\S\ref{sec:contract-promises}) and needs the \emph{stream} only from $s_0 \pre hi_i$. Everything a unit's keys would contribute at or before the high edge is withheld by clause (a) anyway, and the window's content enters through the buffer instead.
That weaker requirement is precisely the handoff condition that the
parallel-chunks instance must name as an assumption
(\S\ref{sec:instances}).

\begin{lemma}[Per-key shape of the gated emission]\label{lem:rm-shape}
Let $k \in D_i$. The entries for $k$ in the range-merge emission
are: if $M_i(k) \neq \perp$, one merged entry with result
$M_i(k)$ at unit $i$'s close. Then $k$'s events in
$\window{hi_i}{f}$, in commit order. Nothing else is emitted, because every earlier
occurrence of $k$ is withheld.
\end{lemma}

\begin{proof}
The gate blocks $k$'s events at or before $hi_i$ from the stream,
and O1
makes unit $i$ the only unit whose close can mention $k$. Post-close
events pass the gate, at their own places and in commit order,
after the close block.
\end{proof}

\begin{theorem}[Range-merge equivalence]\label{thm:rm}
Under clauses (a) and (b), the replay of the range-merge emission
equals the canonical replay:
$\mathrm{replay}(\text{emission})(k) = \sink_f(k)$ for all
$k \in \keys$.
\end{theorem}

\begin{proof}
Fix $k \in D_i$. If some event for $k$ lies in
$\window{hi_i}{f}$, its last such event is the last stream entry
for $k$ and also the last event of $\window{lo_i}{f}$, so both
replays return its image, as in Theorem~\ref{thm:discard}.
Otherwise no event for $k$ lies beyond the high edge, and the
canonical fold collapses to the merged result. If
$\window{lo_i}{hi_i}$ holds an event for $k$, its last one decides
both $\sink_f(k)$ and $M_i(k)$. If not, both equal $R_i(k)$.
Hence $\sink_f(k) = M_i(k)$. On the stream side,
Lemma~\ref{lem:rm-shape} leaves either the single merged entry,
replaying to $M_i(k)$, or, when $M_i(k) = \perp$, no entry at all,
replaying to absent, which is $M_i(k)$ once more.
\end{proof}

\begin{theorem}[Range-merge monotonicity]\label{thm:rm-mono}
With tags as in Theorem~\ref{thm:discard-mono} and the merged
entry tagged with the length of $\den{hi_i}$, each key's tags are
non-decreasing along the
range-merge emission.
\end{theorem}

\begin{proof}
By Lemma~\ref{lem:rm-shape} a key's entries are one merged entry
at its unit's close, accurately tagged by Lemma~\ref{lem:merged},
followed by
events with tags beyond it, in commit order.
\end{proof}

Clause (a)'s full reach is necessary. A natural implementation
deduplicates exactly what it buffered by absorbing the window's
events into the merged unit while streaming every other consumed
event, narrowing withholding to the window alone. That variant
is wrong, and the failure needs only two events.

\begin{proposition}[Withholding is necessary]\label{prop:withhold}
There is a contract instance on one key for which the gated
emission replays correctly, while the window-narrowed variant
(absorb the window into the buffer, stream every other consumed
event) replays the key to a value the source no longer holds.
\end{proposition}

\begin{proof}
Over an initially empty store, take
$L = e_1\, e_2$ with $\keyof(e_1) = a$, $\imgof(e_1) = 5$,
$\keyof(e_2) = a$, $\imgof(e_2) = \perp$, representing an insert followed by a
delete. Coordinates $c_0, c_1, c_2$ map to the prefixes of lengths
$0, 1, 2$. The plan has one unit with domain $\{a\}$, bracket $(c_1, c_2)$,
frontier $f = c_2$, and refresh $R(a) = 5$, reflecting the state at the low edge
$c_1$ where the insert has committed and the delete has not.
Consumption starts at $s_0 = c_0$.

The window $\window{c_1}{c_2}$ contains the delete, so the merged
result is $M(a) = \perp$ and the lawful emission for $a$ is
empty. By Theorem~\ref{thm:rm} it replays $a$ to absent, which is
$\sigma_{\den{c_2}}(a)$. Under the window-narrowed variant, the
in-window delete is absorbed into the buffer, whose merged block
is empty and emits nothing, while the pre-bracket insert $e_1$
lies outside the window and passes to the stream. The variant's
stream is the single entry $e_1$, and it
replays $a$ to $5$, which is a value the source deleted before the
frontier. Clause (a)'s full reach, withholding everything at or before the high edge rather than only what the buffer absorbed, is exactly what separates the correct replay from the wrong one.
\end{proof}

Window-discard and range-merge ship visibly different streams
from the same plan. The first interleaves a unit's in-window
events at their own places, whereas the second absorbs them into
one merged block. Theorems~\ref{thm:discard} and~\ref{thm:rm} say
both replay to the same canonical sink. The log-wins property is proved once for the canonical rule, and each concrete merge discipline then needs only a short equivalence proof.

\subsection{Point-splice}\label{sec:merge-point}

A point bracket has $lo_i = hi_i$, the window is empty, and there
is nothing to reconcile. No key has an in-window event, and the merged
result is the refresh itself. A point splice is the degenerate case of
window-discard, with $s_0 \pre lo_i$ for every unit. Each unit emits its
present refresh entries at its splice coordinate, followed by later
events in log order. A composite with distinct splice coordinates must
therefore begin at or before the earliest splice, because beginning at the
latest would omit events needed by earlier units. Theorem
~\ref{thm:discard} gives replay equivalence, and Corollary
~\ref{cor:emitted-cut} gives final source equality. This is equality
after replay, not entry-for-entry identity. The splice instances of
\S\ref{sec:instances} use this degenerate case.

\subsection{Analyzed and open}\label{sec:merge-open}

Four further disciplines run in production or arise naturally from
the model. For each we state the equivalence obligation in the
contract's terms and record it as \emph{open}. None is established
in this paper, and nothing in this subsection is a result.

Readers evaluating only the five production placements may continue
directly to \S\ref{sec:instances}. This subsection records design
obligations for additional mechanisms.

\paragraph{Key-range gating.}
This discipline uses no buffer at all. The capture maintains a copy cursor over a
totally ordered key domain, applies log events only for keys at or
below the cursor, and discards events for keys ahead of it, on the
grounds that the later copy pass will read those rows anyway. The
obligation has two parts. The plan must be a monotone cursor over
the ordered domain, and each key's refresh witness (O2's
coordinate) must lie at or after every event discarded for that
key while it was ahead of the cursor, so that the eventual read is
fresh under an advancing view. A frozen read view breaks the
second part. Discarding an observed ahead-of-cursor event and then fetching
that row from an older snapshot loses the discarded
change with nothing left to repair it. A non-monotone
visit order breaks the first. Both preconditions are stated here.
Neither implication is proved, and the discipline is open.

\paragraph{Version-max.}
With a version column on every row, a sink can apply refresh
entries and events in any order and keep, per key, the entry of
maximal version, allowing reconciliation to succeed without any bracket. 
Deletes would also need a versioned tombstone or equivalent 
evidence to reject stale copied rows. The obligation is that 
the version order agree with commit order on
every pair of entries for the same key. Engine-generated,
commit-ordered versions can satisfy it, whereas client-supplied
timestamps in general do not, and a deployment leaning on them
sits at T1 (\S\ref{sec:tiers}) until the agreement is
established, so the discipline remains open.

\paragraph{Re-read normalization.}
Under a delta log, events carry changed columns rather than full
images, S-IMG fails, and the core fold is not even well typed.
The industrial repair is to re-select colliding keys under a fresh
bracket and emit the read images. In this model the repair is a
\emph{log normalization}, not an in-core discipline. The repaired
instance is modeled over effective full-image events, each re-read
supplying the image its delta event lacked, and the instance's
equivalence obligation then includes the normalization itself.
That obligation has a termination half, since re-reads race new
deltas and the re-read procedure must be shown to terminate. None of the
vendor documentation surveyed for this paper states one. The discipline
remains open, with the termination half flagged as the substantive part.

\paragraph{Rewind-by-preimage.}
Where the log carries before-images, a late read can be rolled
backward by applying preimages in reverse until the unit's low edge,
yielding a refresh valid at $lo_i$ itself. Its obligation is to 
identify the suffix reflected in the read and 
show that reversing it reconstructs each key at $lo_i$. It is 
recorded as a construction worth having,
but is entirely unproved and remains open.

\section{Instantiating the Contract}\label{sec:instances}

The preceding sections separate three questions. Does a capture plan
satisfy the source-side contract? Does its emitted stream replay like
the canonical merge? Which semantic level and certification status
follow? This section answers
those questions for five production protocol shapes. It is the formal
instantiation layer of the paper. Section~\ref{sec:practitioner}
distills the resulting framework and its operational consequences into
a standalone practitioner synthesis.

\subsection{How a protocol inherits the theorem}
\label{sec:instances-pattern}

Each instantiation follows the same proof pattern.

\begin{enumerate}
\item Map the protocol's coordinates to nested log prefixes consistently
with its actual coordinate behavior, including the exact interpretation
of its low and high edges.
\item Map its copied scopes, reads, brackets, and frontier to a capture
plan, then satisfy the source promises and plan obligations of
Definition~\ref{def:contract}.
\item Match the emitted stream to a proved discipline in
\S\ref{sec:merge}, or state the additional equivalence condition that
the implementation must establish.
\item Apply Theorem~\ref{thm:cut} for the frontier result and
Lemma~\ref{lem:latest-high} for the conservative post-close
trajectory. Apply Corollary~\ref{cor:shared} only when the capture
exhibits one shared state witness. Add transaction closure before
reading either result as a transactional snapshot or history.
\end{enumerate}

The proof is conditional at the points where the protocol
description does not itself establish a source, read, coordinate, or
emission property. Such requirements appear below as explicit preconditions
and theorem assumptions. Release-specific evidence establishes whether a
particular implementation satisfies them.

\subsection{What an instantiation establishes}
\label{sec:instances-meaning}

A contract proof establishes the canonical frontier result. An emitted
stream attains it only after the corresponding merge-equivalence
argument is established. The two claims are kept separate below because
the same plan can be emitted by different buffering and handoff rules.
Likewise, a T3 placement concerns the retrospective canonical family
of Definition~\ref{def:traj} and an exhibited common state witness. It
does not assert that every physical
prefix of an emitted stream follows that family. Those distinctions let
the formal placements remain precise while the practitioner synthesis
can focus on what the framework means and when its results apply.

\subsection{Classic watermarked DBLog}\label{sec:inst-classic}

\paragraph{Construction.}
The original algorithm~\cite{andreakis2020dblog} reads each table
in chunks. It brackets every chunk select between two watermark
writes, updates to a dedicated single-row table whose change events
appear in the tailed log. It takes the log positions of those two
events as the bracket and deduplicates by window-discard
(\S\ref{sec:merge-discard}). Its published
formalization~\cite{andreakis2026virtualcuts} defines wellformed
runs of this shape and proves them correct. That the present framework 
contains the original as an instance is therefore a theorem with 
a precise form. Every wellformed run yields a contract instance, 
and the contract's replay agrees with the published formalization's 
own conclusion.

The import adopts five clauses of published wellformedness,
restated here in this paper's vocabulary:

\begin{itemize}
\item \emph{source binding}: the run's source history is a
committed, commit-ordered event history with coordinates
non-decreasing along it.
\item \emph{chunk partition}: the chunks' key domains are pairwise
disjoint and jointly cover the run's scope.
\item \emph{chunk evidence}: each chunk carries one read
coordinate lying between its lower and upper watermarks, and its
read map is total on its domain.
\item \emph{read validity}: each chunk's read results equal the source's
row-event replay state at that chunk's read coordinate, key by key.
\item \emph{frontier bound}: every read coordinate lies at or
before the run's frontier.
\end{itemize}

\begin{theorem}[Classic DBLog satisfies the contract]\label{thm:classic}
Every wellformed run of the classic algorithm, in the sense
of~\cite{andreakis2026virtualcuts}, determines (i) a contract
instance (Definition~\ref{def:contract}) and (ii) an instance whose
canonical replay at the run's frontier equals the run's own source
row-event replay state there, which is, per key, the image of the last 
history event for the key with coordinate at or before the frontier, 
or the initial value if no such event exists.
\end{theorem}

\begin{proof}
Map the run's objects onto the contract's. The source history maps
to the log. Inserts and updates become events with present images,
while deletes become events with absent images. This is the 2020
all-columns post-image assumption serving as S-IMG, and S-LOG is
the source binding. A history coordinate identifies the prefix of all
events with coordinates at or before it. Since coordinates are
non-decreasing along the history, that filter is a genuine prefix,
and increasing a raw coordinate either leaves the corresponding log 
prefix unchanged or extends it. Different raw coordinates can refer 
to the same prefix.
This at-or-before convention is the instance's declared endpoint
normalization (\S\ref{sec:setting-windows}). The prefix includes every event
at the chosen coordinate, so a run of equal-coordinate events
lands wholly inside it, matching the published formalization's own
latest-lookup convention, which resolves coordinate ties by
position.

Each chunk becomes a unit. Its domain is the chunk domain, its
bracket edges are the lower watermark and the upper watermark
truncated at the frontier, and its refresh is the chunk's read
map. The truncation is needed because the published wellformedness
bounds the \emph{read} coordinate by the frontier while
never bounding the upper watermark, whereas contract
plans live at or before the frontier by definition. Chunk evidence
and the frontier bound keep the read coordinate inside the
truncated bracket, so nothing is lost. The obligations are now
satisfied clause by clause. O1 is the chunk partition. O2 holds
with every key of a chunk witnessed uniformly at the chunk's one
read coordinate, whose validity is the read-validity clause. This is the upward
direction of Proposition~\ref{prop:torn}'s discussion, a single
read point instantiating the per-key existential. O3 is the
totality half of chunk evidence. O5 holds by construction, since
the import maps the same history that generated the change stream,
so both paths share one event vocabulary. This proves (i).

For (ii), the prefix for a coordinate consists exactly of the mapped
events with coordinates at or before that coordinate, so the
state after it is, key by key, the image of the key's last such
event (ties across equal coordinates resolved by position, per
the at-or-before reading declared above), or the initial value if
none exists. Theorem~\ref{thm:cut}
applied to the instance equates the canonical replay at the
frontier with exactly that state.
\end{proof}

The published formalization's observed-stream clauses and its
replay and certificate layers are not needed here. In
this framework the log is the committed history itself, and
observation enters at the discipline layer through S-OBS
(\S\ref{sec:merge}). The two developments meet at their
conclusions. The state in Theorem~\ref{thm:classic}(ii) is
precisely the state the published formalization's own replay
theorem computes at the frontier, the same per-key latest
lookup under the same position-resolved tie
convention~\cite{andreakis2026virtualcuts}. The
contract-level theorem thereby generalizes the published result
and never contradicts it.

Theorem~\ref{thm:classic} imports the plan-level T2 result into the
generalized framework. The earlier formalization separately proves a
clean-prefix theorem for the emitted stream in its own run semantics
~\cite{andreakis2026virtualcuts}. A physical generalized
window-discard conclusion needs an additional bridge. It requires
S-OBS, the survivor and close behavior of
Theorem~\ref{thm:discard}, and every real upper marker at or before the
frontier being claimed. The last condition matters because the import
replaces an upper marker beyond $f$ by $f$ in the canonical plan. A
formal close at that truncated edge is not automatically a physical
close in the original run. One may instead choose a later frontier
that includes the real marker. This is a statement about where the proof stops. It says nothing
about the correctness of classic DBLog code. Corollary
~\ref{cor:emitted-cut} gives the final emitted equality once the bridge
holds. None of these results claims that each prefix of the emitted
stream follows the canonical trajectory of
Definition~\ref{def:traj}.

\paragraph{Watermark preconditions (W1-W6).}
The formal import assumes wellformedness. A deployment must also
be \emph{able} to run the bracket mechanism. The
watermark-write design carries six preconditions. We make them explicit
here because \S\ref{sec:audit} shows the 2020 paper's own support
conditions omit W1-W5 and state W6 only as non-stale reads:

\begin{description}
\item[\textbf{W1}] \emph{provisioning}: the watermark table exists
in a dedicated namespace, or the capture role holds DDL privilege
to create it.
\item[\textbf{W2}] \emph{write access}: the capture role can
update it, twice per chunk, for the life of every run. Read-only
sources, read-only roles, and frozen archives are excluded.
\item[\textbf{W3}] \emph{capture coverage}: the watermark table's
changes appear in the same tailed stream as the data, a configuration 
property of publications and log filters. A filtered-out low watermark 
prevents the protocol from observing the window opening, and without an 
external timeout or recovery mechanism the run can wait indefinitely.
\item[\textbf{W4}] \emph{recognizability}: the watermark value
column is among the captured columns, so any column filtering must retain that value.
\item[\textbf{W5}] \emph{ordering coupling}: watermark writes,
chunk reads, and the tailed log target the same instance, one
linear commit history. Using a primary for writes and a replica 
for reads requires a separate argument that the read lies within 
the watermark bracket in the same history.
\item[\textbf{W6}] \emph{visibility order}: a chunk read that starts after
the low watermark write has completed reflects every transaction whose
events precede that watermark in the log. MySQL documents that, with
\texttt{binlog\_\allowbreak{}order\_\allowbreak{}commits} enabled (the
default), transactions are committed to the storage engine in the order
they are written to the binary log~\cite{mysql_binlog_options}.
\end{description}

\paragraph{W6 on a PostgreSQL primary.}
A PostgreSQL primary does not guarantee W6. A committing transaction
first writes its commit record to the log. In a later, separate step the
server stops treating the transaction as in progress~\cite{postgres_commit_source,aws_pg_visibility}.
For concurrent transactions, this step can occur in a different order from
that of their commit records~\cite{aws_pg_visibility}. A statement snapshot can
therefore contain one transaction but lack another logged before it, so the
transactions in the snapshot do not form a prefix of the log.
Logical decoding delivers a transaction once its commit record
is decoded, in the order of the commit
records~\cite{postgres_logicaldecoding_plugins,postgres_commit_source}. It
can therefore deliver a transaction that new snapshots still treat as in
progress.

Suppose the commit record of a transaction precedes the low watermark in the
log, and the transaction is still in progress when the chunk read starts.
Suppose also that the transaction changes a key from $100$ to $50$ or deletes
it, and that no later event before the high watermark changes that key. The
read then returns $100$, so O2 fails for the key. Because no event for the key
lies in the window, window-discard keeps the buffered row and releases it at
the close, after the transaction's own event. The stale value $100$ stays until
the key changes again, and a row that the transaction deleted is resurrected.

The capture can restore O2 for the chunk read with the snapshot of that read
itself. It obtains that snapshot with
\texttt{pg\_\allowbreak{}current\_\allowbreak{}snapshot()} in the read
statement~\cite{postgres_commit_source}, or in one repeatable-read
transaction that also runs the read~\cite{postgres_transaction_iso}. A
logged transaction is outside that snapshot exactly when its transaction
ID is at or above \texttt{xmax} or appears in the snapshot's list of transactions in
progress~\cite{postgres_snapshot_functions}. When the first logged
transaction outside the snapshot precedes the low watermark, the capture
moves the low edge back to immediately before that transaction or to an
earlier transaction boundary. It then applies window-discard over the wider
window. We call this the
\emph{snapshot comparison}. It has three conditions:

\begin{itemize}
\item Streamed transaction IDs and snapshot IDs are compared in the same
ID space.
\item Every event of the wider window can remove a buffered row, including
events the consumer processed before the read. The capture therefore
retains their keys. Alternatively, it repeats the read with a fresh
snapshot until every transaction logged before the low watermark is inside
that snapshot, and only then writes the high watermark. A third way
applies to a capture that pauses log processing during the read, as the
original algorithm does~\cite{andreakis2020dblog}. The pause begins between
two transactions, before the low watermark write, and the capture places the
low edge where it paused. Window-discard then starts at the pause position
rather than at the low watermark. It accepts a read only if every transaction
before that position is inside the read's snapshot. It writes the high
watermark only after an accepted read. It can check this condition directly
for the transactions it has processed. Earlier transactions, for example
those before a restart, need a separate check. It keeps no keys of processed
events, because it processes every event of the wider window after the read.
\item For each key, take the logged transactions inside the snapshot that
changed the key. The read returns the value left by the last of them in log
order, or the key's initial value if there are none.
\end{itemize}

We assume the third condition and do not derive it from the commit
procedure. It corresponds to PostgreSQL's documented row-level locking,
under which an update of a row already updated by a concurrent transaction
waits until that transaction commits or rolls
back~\cite{postgres_transaction_iso}.

Under these conditions, every logged transaction before the low edge of
the wider bracket is inside the snapshot. Every transaction inside the
snapshot had completed before the read started, so it precedes the
high watermark in the log. Now fix one key. By the third condition, if
the last logged transaction inside the snapshot that changed the key lies beyond
the low edge, the read result is the replay state immediately after that
transaction.
Otherwise the read result is the replay state
at the low edge. In both cases the coordinate lies inside the wider
bracket, so O2 holds there. Theorems~\ref{thm:cut} and~\ref{thm:discard} then apply
to the wider window.

\paragraph{Result.} The imported plan is T2, by
Theorem~\ref{thm:classic} and Theorem~\ref{thm:cut}. A
frontier-replay result for the emitted stream additionally depends on the physical-marker,
observation, and window-discard requirements stated above.
The certified configuration of the
published formalization additionally carries an externally
checkable certificate of its cut, which is the T4 certification
overlay (\S\ref{sec:tiers-rungs}). The classic instance is where the
framework incorporates its T4 anchor~\cite{andreakis2026virtualcuts}.

\paragraph{Debezium's signal-table modes.}
In the documented non-read-only mode, Debezium's default
\texttt{insert\_insert} strategy writes opening and closing records to
the configured signaling data collection and reconciles the intervening
row events by window-discard
~\cite{debezium_signalling,debezium_mysql_connector}. This has the same
span-bracket and log-wins shape as the classic placement, under four
conditions. W1-W6 hold for the signaling collection. The opening
record precedes the chunk read, and the closing record follows it. The
scan returns one value-or-absence result for every owned key, each valid
at some coordinate inside the bracket (O2-O3). Release waits until the
real close and all relevant window events have been observed. Stable scope,
identity, and representation agreement remain required.

Under those conditions, Theorem~\ref{thm:cut} places the plan at T2.
Theorem~\ref{thm:discard} and Corollary~\ref{cor:emitted-cut} give the
finite emitted output final source replay. This bridges Debezium's documented mode to the watermarked placement
under the stated conditions. It does not apply Theorem~\ref{thm:classic}'s
imported wellformedness to a Debezium release. For the PostgreSQL
connector's signal-table mode~\cite{debezium_postgresql_connector}, the
same argument holds with the snapshot comparison in place of W6.

The \texttt{insert\_delete} strategy uses the same bracket. The opening
insert supplies the low edge, and deleting that row supplies the high edge.
In the tagged source, Debezium accepts the delete as a signaling event,
reads the opening row's identifier from its before-image, and converts
it to the same internal window-close signal used by
\texttt{insert\_insert}~\cite{debezium_signal_delete_source}. This
placement therefore adds the requirement that the delete event's
before-image retain the opening row's identifier (the signal table's
primary key). Without it, the connector cannot recognize the close.
Enabling the source signaling channel also does not itself establish capture coverage. The
signaling collection's changes must actually be present and recognizable
in the consumed history.

\subsection{Debezium read-only capture: MySQL, MariaDB, and PostgreSQL}
\label{sec:inst-readonly}

\paragraph{Coordinate construction.}

Debezium's incremental snapshots adapt DBLog's watermark-based
design~\cite{debezium_incremental_snapshots,debezium_ddd3}. Its MySQL and MariaDB
read-only connectors bracket chunks by reading
server GTID state instead of writing marker rows
~\cite{debezium_mysql_connector,debezium_mariadb_connector,debezium_readonly_blog}.
Its PostgreSQL read-only connector brackets them with two samples of
the server's transaction snapshot
~\cite{debezium_ddd8,debezium_postgresql_connector}. The
trigger can also arrive over a write-free channel
~\cite{debezium_signalling}. The three servers do not expose the same
coordinate, however. MySQL reports a full executed-GTID set. MariaDB's
\texttt{GTID\_\allowbreak{}BINLOG\_\allowbreak{}POS} reports one frontier per replication domain.
PostgreSQL returns a lower transaction-ID bound (\texttt{xmin}), an upper
bound (\texttt{xmax}), and a list of transactions still in progress.
We therefore prove the set-oracle construction once, then give
separate product bridges below.

The generic construction applies to a coordinate that is the exact
set of transactions contained in one transaction-complete prefix of
the consumed history. The difference between high and low then selects
the complete transaction blocks that entered the prefix while the
chunk was being captured. This is stronger and more precise than
calling the sample atomic. A sample must end between complete
transaction blocks and be gap-free relative to that same consumed
history.

Model the history as follows. Positions $p$ of the log carry a
transaction label $\tau(p)$, and one assumption shapes the
labeling. Equal labels occupy contiguous position blocks, because
a transaction commits as one uninterrupted run of row events. Call
a position $n$ a \emph{boundary} when no transaction straddles it, meaning
no position before $n$ shares a label with any position at or
after $n$. After projection onto labels represented in the scoped log,
every accepted sample must equal the complete transaction-label set of
the prefix at one such boundary. The \emph{executed set} at boundary
$n$ is
\[
G(n) \;=\; \{\, \tau(p) : p < n \,\}.
\]
An operational executed set can also contain transactions that
have no row event in the capture scope. The model uses the
projection of that set onto transaction labels represented in the
scoped log. Applying the result therefore requires projection
agreement. For every captured event, membership of its label in
the raw low and high sets must agree with membership in the
projected sets $G(n)$. Extra labels are harmless because the
window test asks only about the captured event's own label. This
bridge is included in A-linear below.

The set $G(n)$ maps to the length-$n$ prefix. The mapping is
well defined because distinct boundaries carry distinct executed
sets. If $m < n$ are boundaries, the label $\tau(m)$ lies in
$G(n)$ but not in $G(m)$, since membership in $G(m)$ would put an
earlier position and position $m$ on one label across the boundary
$m$. The corresponding prefixes are therefore nested, and inclusion
between executed sets agrees with prefix inclusion, so
Definition~\ref{def:coords} is satisfied without any scalar.

What remains is admissibility's fidelity requirement 
(\S\ref{sec:setting-coords}), under which the instance's own membership
decision must agree with the window defined by the mapped prefixes. That is a
theorem.

\begin{theorem}[Transaction-complete executed-set windows are faithful]
\label{thm:oracle}
Let $n_{lo} \le n_{hi}$ be boundaries and $p$ any log position.
Then
\[
\tau(p) \in G(n_{hi}) \setminus G(n_{lo})
\quad\Longleftrightarrow\quad
n_{lo} \le p < n_{hi}.
\]
The operational test, ``the event's transaction identifier lies in
the high set and not in the low set,'' decides exactly the
half-open window of log events.
\end{theorem}

\begin{proof}
Suppose $\tau(p) \in G(n_{hi}) \setminus G(n_{lo})$. If $p$ lay
below $n_{lo}$, its label would be in $G(n_{lo})$ by definition,
so $n_{lo} \le p$. If $p$ lay at or beyond $n_{hi}$, then
membership in $G(n_{hi})$ would give an earlier position $q <
n_{hi}$ with $\tau(q) = \tau(p)$, and $q$ and $p$ would straddle
the boundary $n_{hi}$. Hence $p < n_{hi}$. Conversely let
$n_{lo} \le p < n_{hi}$. Then $\tau(p) \in G(n_{hi})$ directly,
and $\tau(p) \notin G(n_{lo})$ because an earlier position $q < n_{lo}$
with the same label would straddle the boundary $n_{lo}$.
\end{proof}

Theorem~\ref{thm:oracle} establishes the coordinate-fidelity requirement 
of the read-only placement. The remaining work is to check the common
contract and the buffer-and-discard emission.

\begin{proposition}[Placement: read-only]\label{prop:readonly}
Consider a plan whose units carry exact transaction-complete
executed-set brackets, whose final domains partition the scope, and
whose high boundaries lie at or below the frontier. Assume A-linear,
A-scope, A-read, and A-representation as stated below. Then:
\begin{enumerate}
\item the plan satisfies the contract, and its canonical replay is T2
by Theorem~\ref{thm:cut}
\item a finite buffer-and-discard output has the same final
source-side replay under A-discard and the window-discard form of
S-OBS, by Theorem~\ref{thm:discard} and
Corollary~\ref{cor:emitted-cut}.
\end{enumerate}
\end{proposition}

\begin{proof}
Theorem~\ref{thm:oracle} gives coordinate fidelity. A-scope gives O1
and O3. A-read gives O2. A-representation gives S-IMG and O5, while
A-linear binds every object to S-LOG. The bracket and frontier bounds
complete the contract, so Theorem~\ref{thm:cut} proves part 1. Under
A-discard and S-OBS, the emitted stream realizes the window-discard
construction. Theorem~\ref{thm:discard} proves emission equivalence,
and Corollary~\ref{cor:emitted-cut} composes it with the cut for part
2.
\end{proof}

\paragraph{Read-only assumptions.}
\begin{description}
\RaggedRight
\item[\textbf{A-linear}] \emph{history, prefix, and projection}:
samples, chunk reads, and streamed events refer to one unchanged
commit-ordered history whose transaction identifiers occupy contiguous
blocks. Each accepted sample, after the declared projection, equals the
gap-free set of complete transaction blocks in one prefix of that
history. For every captured row event, membership of its transaction
identifier in the raw sample agrees with membership after projection to
the modeled scope. These are three separate clauses. One cannot be
inferred from either of the others.
\item[\textbf{A-scope}] \emph{complete stable ownership}: the final
unit domains partition the captured identity namespace and remain
stable for the run. The scope witness is coordinate-bound. Each chunk
scan supplies complete membership, so a nonreturned owned identity is
a justified absence rather than an unobserved row.
\item[\textbf{A-read}] \emph{per-key bracket match}: O2 at
instance level. Every returned value or absence must match the source
at some transaction-boundary prefix inside the bracket. This condition 
is kept explicit because the connector documentation does not fully
specify how chunk reads account for every key and ensure that each
returned value matches source state within the sampled interval
~\cite{debezium_mysql_connector,debezium_mariadb_connector}.
\item[\textbf{A-representation}] \emph{image and identity
agreement}: the streamed and copied paths share the same stable key
identity and value encoding. Effective events carry complete
post-images, deletes represent absence, and a primary-key change is
normalized as old-key absence plus new-key value. This satisfies
S-IMG and O5.
\item[\textbf{A-discard}] \emph{serialized window lifecycle}: the low
sample is installed before the read and the high sample is taken after
the buffer is complete, or an equivalent serialization establishes
that order. Every relevant row event whose transaction lies in the
high-minus-low window can remove its buffered key. The buffer is
released only after all row events of the last window transaction have
been processed. The first later transaction may establish closure
before processing its own row events because those events lie outside
the window. A heartbeat may close only after all window events are
known to have been processed.
\item[\textbf{S-OBS}] \emph{complete advancement and coverage}: one
commit-ordered consumed stream begins at or before every low boundary
and continues through the frontier. Filtered and out-of-scope events
needed to advance the window state remain observable even when they
do not enter the emitted data stream. Retention covers every range the
discipline reads.
\end{description}

\paragraph{MySQL bridge.}
Debezium 3.6.1.Final samples MySQL's
\texttt{Executed\_Gtid\_Set} through binary-log status before and
after a chunk read, subtracts the low set from the high set, and uses
event membership to manage the in-memory window
~\cite{debezium_mysql_connector,debezium_mysql_readonly_source}.
This is the direct full-set bridge to Theorem~\ref{thm:oracle} under
A-linear. MySQL documents that a committed GTID is externalized into
\texttt{gtid\_executed} non-atomically shortly after commit
~\cite{mysql_gtid_lifecycle}. The bridge therefore does not assume
atomic externalization. It assumes the stronger semantic fact the
proof needs, namely that each accepted value maps to an exact complete
prefix of the same history used by the read and consumed log. GTID mode
and GTID consistency are required. On a multithreaded MySQL replica,
the documented commit-order-preservation setting is also needed for
this one-history mapping~\cite{debezium_mysql_connector,mysql_gtid_lifecycle}.

\paragraph{MariaDB bridge.}
Debezium 3.6.1.Final instead samples
\texttt{GTID\_\allowbreak{}BINLOG\_\allowbreak{}POS}. MariaDB defines this value as the last
event group written to the binary log in each replication domain, so
it is a vector of per-domain frontiers rather than MySQL's complete
executed set~\cite{mariadb_gtid}. The tagged connector samples that
value. A nonempty high-minus-low GTID difference causes the shared
binlog source to request a reread. The common reread helper also
requires a running snapshot, an open deduplication window, and a
nonempty buffer, so the request alone does not establish an effective
retry
~\cite{debezium_mariadb_readonly_source,debezium_mysql_readonly_source}.
Under the stated assumptions, we conclude correctness in two cases. If the high-minus-low
difference is empty, the same-history and monotone-prefix
assumptions make the complete vector a confirmed quiet point in the exact
consumed MariaDB binlog. The chunk is then a point splice, provided
A-read supplies complete value-or-absence results and the vector
maps consistently to the same prefix of the consumed log across every relevant domain.

If the difference is nonempty and the buffer is empty, the reread
helper returns without another scan. This case needs no quiet-point
argument. Under A-read and A-scope, every refresh
is a justified absence, so the survivor set is empty. The same-history and 
monotone-prefix assumptions still ensure that the two samples define a well-formed 
bracket. Under A-discard and S-OBS, every racing or later event,
including an insert that raced the empty scan, remains on the emitted
log path. Theorem~\ref{thm:discard} therefore gives final source
equality. This argument uses the samples as bracket bounds. It does not
invoke the MySQL executed-set membership theorem or MySQL-only replica
variables. These two cases do not exhaust the tagged implementation.
For a populated buffer with a nonempty difference, we would need to
show separately that the sampled coordinates identify the correct
interval around the read in the committed log and that replaying the
connector's output reconstructs the source state. Alternatively, we
would need evidence that a reread actually occurs and brings the chunk
into one of the covered cases. MariaDB
deployments with independently advancing domains must state the same-history
and prefix-preserving translation explicitly
~\cite{mariadb_parallel_replication}.

Heartbeat availability remains an operational liveness requirement for detecting that 
the stream advanced. It does not replace A-linear, S-OBS, or the close
ordering in A-discard.

\paragraph{PostgreSQL bridge.}
Before and after each chunk read, Debezium 3.6.1.Final calls
\texttt{pg\_\allowbreak{}current\_\allowbreak{}snapshot()}~\cite{debezium_postgres_readonly_source}. Debezium's
design document for this mode describes the approach~\cite{debezium_ddd8}.
The mode requires PostgreSQL 13 or later~\cite{debezium_postgresql_connector}. A sample names a lower bound
\texttt{xmin}, an upper bound \texttt{xmax}, and the top-level
transactions between the two that were still in progress. Every
top-level transaction below \texttt{xmax} and outside that list had
completed when the sample was taken, while transactions at or above
\texttt{xmax} had not yet completed~\cite{postgres_snapshot_functions}. A sample
therefore describes a set of completed transactions, as an executed-GTID
set does for MySQL. The IDs of transactions still in progress
only delimit that set. The construction reads no event of an active
transaction and no dirty state, so it stays within the committed history
of \S\ref{sec:setting-store}. A-linear for this server requires each
sample, projected to the logged transactions, to equal the transaction
set of one complete prefix of the consumed history.
Theorem~\ref{thm:oracle} then applies unchanged to membership in the
high set and not in the low set.

On the data-event path, the connector compares transaction IDs with the
sampled bounds to open and close the window. The first event whose
transaction ID is at least the low sample's \texttt{xmin} opens it. While
the window is open, the connector drops buffered rows whose keys appear
in the log. The first event whose transaction ID exceeds both sampled
\texttt{xmax} values closes it. The surviving rows are emitted before
that event~\cite{debezium_postgres_readonly_source,debezium_signal_delete_source}.
For this window, A-discard requires the connector to retain the buffer
until the closing comparison succeeds and allow every event in the
window to remove a buffered row with the same key.

The comparisons require the streamed transaction IDs and sampled bounds
to use the same ID space. When the low sample takes effect, it
must include every transaction the consumer has already processed. As
with MySQL's \texttt{Executed\_Gtid\_Set}, we assume that each accepted
sample represents exactly the transactions in one complete prefix of the
consumed log. A PostgreSQL primary guarantees neither of the last two
conditions (\S\ref{sec:inst-classic}), and we keep both as assumptions.

Under these conditions, the connector's window contains the interval
between the samples. Every event at or beyond the sampled low boundary
remains to be processed when the low sample takes effect. Its transaction
is outside the low set, so its ID is at least the low \texttt{xmin}. The
opening comparison therefore succeeds by the first such event. The low
edge, immediately before the opening event, lies at or before the
sampled low boundary. The closing event has an ID greater than both
sampled \texttt{xmax} values, so its transaction is outside the high set.
Because that set represents a log prefix, the closing event is after
that prefix. The high edge, immediately before the closing event,
therefore lies at or after the sampled high boundary.

Membership in the connector's window is determined by log position. This
window contains every read coordinate from the sampled bracket. The
observation frontier must be at or beyond its high edge. Widening the
bracket preserves O2 (\S\ref{sec:theorem-cut}), and the
window-discard result applies under the conditions stated above.

\paragraph{Result.}
Under their respective bridges and the shared assumptions, the
MySQL set plan, the two analyzed MariaDB cases, and the PostgreSQL plan are T2.
The empty-difference case uses a point splice, the empty-buffer,
nonempty-difference case uses window-discard over the sampled bracket,
and the PostgreSQL plan uses window-discard over the connector's window.
Their actual finite outputs have the same final source replay only under
the lifecycle, observation, and emission-equivalence conditions of the
corresponding discipline. Write-freeness is visible in these
constructions because none of Proposition~\ref{prop:readonly}'s assumptions places a marker event in the log.

\subsection{Flink CDC parallel chunks}\label{sec:inst-parallel}

\paragraph{Protocol mapping.}
The protocol splits a table into snapshot chunks that may run in
parallel. The release-3.6.0 source records a low binlog offset, runs the
chunk select, records a high offset, reads the corresponding binlog
range, and reconciles those events into the buffered chunk before
handoff~\cite{flinkcdc_mysql,flinkcdc_source_code}. This is the
range-merge discipline of \S\ref{sec:merge-range}. The plan-level
mapping and the properties of the stream actually emitted by the
connector are explicitly separated below.

Offsets map to prefixes of one source history. The documented
$[\mathit{LOW},\mathit{HIGH}]$ notation does not by itself fix the raw
endpoint behavior. A-edge therefore maps the connector's raw offset
tests to the half-open window of log events
$\den{\mathit{HIGH}}\setminus\den{\mathit{LOW}}$. Comparisons and
membership tests must agree with that mapping.

\begin{proposition}[Placement: parallel chunks]\label{prop:parallel}
Consider a finite final plan of logical chunk units whose offset
brackets have high edges at or below the frontier. Under A-source,
A-read, A-identity, A-edge, A-normalize, and A-representation:
\begin{enumerate}
\item if the final logical domains partition the captured scope, the
plan satisfies the contract and its canonical replay is T2 by
Theorem~\ref{thm:cut}
\item if A-window, A-merge, S-OBS, and A-handoff also hold, the finite
emitted stream has final source-side replay equality by
Theorem~\ref{thm:rm} and Corollary~\ref{cor:emitted-cut}
\item if the actual normalized and gated per-key sequence is the
sequence modeled by Definition~\ref{def:rm}, its source tags are
monotone by Theorem~\ref{thm:rm-mono}.
\end{enumerate}
\end{proposition}

\begin{proof}
For part 1, the explicit partition hypothesis gives O1. A-read gives O2
and O3. A-identity supplies stable logical ownership and final
assignment. A-edge supplies bracket order and coordinate fidelity.
A-source, A-normalize, and A-representation supply the standing history,
image, and representation conditions, including O5. The plan-level proof requires 
no handoff condition. For part 2, A-window and A-merge identify the finite output
with the gated range-merge construction, S-OBS supplies the needed log
coverage, and A-handoff starts the stream at or before every high edge
and supplies clause (a). Theorem~\ref{thm:rm} gives canonical replay
equivalence, which Corollary~\ref{cor:emitted-cut} lifts to source
equality. Part 3 is Theorem~\ref{thm:rm-mono} applied to that actual
normalized sequence.
\end{proof}

\paragraph{Independence from bracket ordering.}
The plan proof relates distinct units through final logical ownership
and one shared frontier bound. It does not order one chunk's bracket
against another's, so brackets may overlap or nest. This is narrower
than saying that the implementation needs no coordination. Split
assignment, checkpointing, phase transition, per-split event filtering,
and recovery still coordinate the running connector. The emitted-stream
result also needs the global handoff and gate. What the theorem removes
is only an additional pairwise bracket-order assumption. A
conforming plan over a disjoint scope composes by
Corollary~\ref{cor:mix}.

\paragraph{Parallel-chunk assumptions.}
\begin{description}
\RaggedRight
\item[\textbf{A-source}] \emph{one history and scope}: offset samples,
chunk reads, backfill records, and the continuing reader refer to one
prefix-preserving source history and one coordinate-bound scope.
\item[\textbf{A-read}] \emph{per-key bracket match}: assumed,
not derived, because the chunk select's isolation and membership
guarantees are not documented here~\cite{flinkcdc_mysql}. Every
returned value or absence must be valid at some coordinate inside the bracket,
and the scan must give a complete read result for its logical domain.
\item[\textbf{A-identity}] \emph{stable logical ownership}: stable
logical identities have one final logical-domain assignment for the
run. A primary-key change is normalized as absence for the old key plus
a value for the new key. For tables with primary keys, the tagged source
uses the record key as the merge-buffer identity. Split routing and the
live per-split gate instead extract the configured chunk-key value from
the appropriate before or after record. A mutable non-key split column
can therefore move a physical row between split predicates and is a
routing hazard, so that case
needs implementation-specific evidence for stable final ownership or
an equivalent proof. Immutability is a useful sufficient condition
where the documentation requires it, not the framework's universal
identity rule
~\cite{flinkcdc_mysql,flinkcdc_source_code}. Flink also supports tables
without a primary key when a non-null chunk column is configured. The
tagged implementation then uses the row image, rather than a stable
logical key, to identify buffered rows. This documented mode is not
placed here unless a deployment supplies the stable-identity
normalization required by \S\ref{sec:setting-store}. The documentation
also reduces the guarantee to at-least-once if the chunk column changes
~\cite{flinkcdc_mysql,flinkcdc_source_code}.
\item[\textbf{A-edge}] \emph{half-open endpoint normalization}: raw
offset inclusion, comparison, and filtering implement
$\den{\mathit{HIGH}}\setminus\den{\mathit{LOW}}$ without a gap or a
double-applied boundary event.
\item[\textbf{A-normalize}] \emph{row-operation normalization}:
Debezium \texttt{READ} and \texttt{CREATE} records become value writes,
\texttt{DELETE} becomes absence, and \texttt{UPDATE} is interpreted as
its before and after pair or as its complete after-image as appropriate.
The release-3.6.0 default deserializer preserves this operation meaning
in Flink \texttt{RowKind} values~\cite{flinkcdc_source_code}.
\item[\textbf{A-representation}] \emph{image and encoding agreement}:
A-normalize yields S-IMG, and copied and logged paths agree on stable
keys and values as required by O5. A custom deserializer must be shown
equivalent. Append-only conversion is usable only when the original
operation metadata is retained or otherwise interpreted before
\texttt{RowKind} is erased. Treating every change as an independent
insert does not by itself represent updates and deletes
~\cite{flinkcdc_source_code}.
\item[\textbf{A-window}] \emph{complete window availability}: for each
chunk, the backfill reader can consume every scoped row event in the
normalized low-to-high window. Retention covers that range.
\item[\textbf{A-merge}] \emph{faithful merge realization}: the
backfill updates each buffered key to the state at the high edge,
including absence, and the emitted chunk represents that state.
\item[\textbf{S-OBS}] \emph{complete observation}: every event needed
to build and close a chunk window, including filtered events used for
position advancement, remains observable through the claimed frontier.
\item[\textbf{A-handoff}] \emph{live floor, gate, and order}: the
continuing reader starts at or before every chunk high. For a chunk's
logical keys, events at or before its high are withheld because their
effects are already in the merged result, and events after its high
are released in source order. The release-3.6.0 assigner records
per-split highs and starts the binlog split at their minimum, so the
reader may encounter an event that precedes another split's low edge.
For each data event, the reader identifies its owning split from the
chunk key and emits the event only when its position is strictly greater
than that split's high watermark~\cite{flinkcdc_source_code}. It
therefore withholds every event at or before the high edge, including
pre-low events, as the range-merge discipline requires.
\end{description}

The official output contract supports the default changelog when
A-normalize and A-representation hold. A custom deserializer is covered
only after the same equivalence is established. An append-only mode that
erases operation kind cannot be used as a value-or-absence stream unless
the lost operation meaning is retained or reconstructed elsewhere.
Checkpoint restart, rescaling, source failover, and end-to-end
exactly-once sink application remain outside this placement. They must
satisfy the plan and emission conditions to obtain the result.

\paragraph{Skipping the backfill.}
The connector offers a mode that skips the per-chunk window read
entirely, deferring in-window changes to the later streaming
phase. Its documentation states that skipping ``might lead to data
inconsistency'' and that ``only at-least-once semantic is
promised''~\cite{flinkcdc_mysql}. The documented behavior does not
satisfy the requirements needed to prove merge equivalence
for this mode. We therefore formulate no T2 mapping for it. The
documentation's stated at-least-once semantics remains unchanged, and
this boundary is not a finding that the implementation is incorrect.
A separate convergence argument could support a different result
(\S\ref{sec:tiers-rungs}).

\paragraph{Result.}
Under its plan conditions, every final logical chunk plan is T2 without
A-handoff. Under the additional window, merge, observation, and handoff
conditions, its finite output has the same final source replay.
Per-key monotonicity additionally concerns the actual normalized gated
sequence and does not follow from final equality alone.

\subsection{Coordinate-bound reads: LSNs, binlog positions,
and consistent points}\label{sec:inst-splice}

\paragraph{Point-splice construction.}

Several engines bind one consistent read to one published log
coordinate, collapsing the bracket to a point. Under
\texttt{REPEATABLE READ}, MariaDB's
\texttt{START TRANSACTION WITH CONSISTENT SNAPSHOT} establishes an
InnoDB read view whose binlog coordinate the session then reads
from the status pair
\texttt{Binlog\_\allowbreak{}snapshot\_\allowbreak{}file}/\allowbreak
\texttt{Binlog\_\allowbreak{}snapshot\_\allowbreak{}position},
documented as ``the binlog position that corresponds to the
snapshot'' and queryable ``in a transactionally consistent way,
irrespective of which other transactions have been committed since
the snapshot was taken''~\cite{mariadb_consistent_snapshot,mariadb_binlog_status_vars,mariadb_snapshot_coordinates}.
Percona Server ports the same
pair~\cite{percona_consistent_snapshot}. PostgreSQL's logical
replication slot creation exports a snapshot together with a
\texttt{consistent\_point}, ``exactly the state of the database
after which all changes will be included in the change
stream''~\cite{postgres_protocol_replication,postgres_logicaldecoding}.
The PostgreSQL instance relies on an upstream fix to slot creation that
PostgreSQL 18.6 does not contain~\cite{postgres_initial_snapshot_fix}.
In each case a select, or a whole-table copy, runs inside that
read view, and log consumption starts at the recorded coordinate. For
a composite of several such units it starts at or before the earliest
of their coordinates (\S\ref{sec:merge-point}).

All these members, and the native dumps of
\S\ref{sec:inst-dump}, instantiate \emph{one} formal object,
because the plan does not record how a read was obtained
(\S\ref{sec:contract-plans}). The unit carries a point
bracket at a recorded coordinate $r$, with the refresh an
engine-consistent read at an engine-internal view coordinate $v$,
published as $p$. The instances differ in the source of their
assumptions and in their operational preconditions, not in their
formal plan.

\begin{proposition}[Placement: the trusted point
splice]\label{prop:splice}
Consider the one-unit plan over dumped scope $S$ with bracket
$lo = hi = r$ and refresh $R$, at frontier $f$ with $r \pre f$.
Assume
\begin{description}
\item[\textbf{A-engine}] $R(k) = \sigma_{\den{v}}(k)$ for every
engine-backed key $k$: the read is consistent at the view
coordinate, for keys of engines that participate in the snapshot.
\item[\textbf{A-engine-scope}] every key of $S$ is engine-backed.
\item[\textbf{A-cert}] $v = p$: the view and the published
coordinate coincide.
\item[\textbf{A-scope}] $r = p$: the recorded coordinate is the
published one.
\end{description}
Then the plan satisfies the contract, and $r$ is a shared state
witness exhibited by the view-coordinate binding, so
Corollary~\ref{cor:shared} applies, ensuring the canonical replay tracks the
source's row-event replay in the canonical trajectory
from $r$ to $f$. The instance sits at T3 on the dumped scope, with an
exhibited shared-state anchor at the recorded splice coordinate. For
this one point unit, $r$ is also the latest high, so T3 adds the exact
shared-state evidence rather than a longer trajectory interval.
Calling that family a transactional history additionally requires
transaction closure throughout the claimed range. If a finite physical
output emits an exact copy of the scope before the observed suffix from
$r$, or otherwise has the same replay, Corollary~\ref{cor:emitted-cut}
also gives final source-side equality at $f$. This is replay
equivalence, not entry-for-entry identity with a canonical stream.
\end{proposition}

\begin{proof}
The assumptions chain into exact point equality at the recorded
coordinate, such that for $k \in S$,
$R(k) = \sigma_{\den{v}}(k) = \sigma_{\den{p}}(k) =
\sigma_{\den{r}}(k)$, using A-engine with A-engine-scope, then
A-cert, then A-scope. O1 holds with the single unit. O2 holds
with the point witness $r$ for every key. O3 and O5 hold at the
typed refresh. The window of the point bracket is empty, so the
discipline is the point splice of \S\ref{sec:merge-point}, and
$r$ satisfies both hypotheses of Corollary~\ref{cor:shared}, because it
lies in the (degenerate) bracket of the only unit and every
refresh matches the source state at it.
\end{proof}

\paragraph{Exact-read assumptions.}
\begin{description}
\RaggedRight
\item[\textbf{A-cert}] the engines' documentation asserts the
view-coordinate coincidence without exposing a mechanism to
check it. O4 records the claim as an explicit assumption, not as a
proved fact. Despite the assumption's historical name, it is a
trusted binding, not the independently checkable certificate that
earns T4.
\item[\textbf{A-scope}] the recorded MariaDB status pair must identify 
the read view used by the copy, captured from the active snapshot 
session before that view is ended or replaced~\cite{mariadb_binlog_status_vars,mariadb_snapshot_session_source}.

\item[\textbf{A-engine}, \textbf{A-engine-scope}] consistency is
engine-scoped (\S\ref{sec:inst-dump} quotes the InnoDB-only
caveat). The dumped scope must lie inside the participating
engines' keys.
\end{description}

S-OBS acquires instance-specific content. The PostgreSQL exported
snapshot is valid only ``until a new command is executed on this
connection or the replication connection is
closed''~\cite{postgres_protocol_replication}, so handle lifetime
is a retention condition on the snapshot side, just as log
retention is on the log side. Oracle documents a coordinate-bound
read primitive, \texttt{AS OF SCN}~\cite{oracle_flashback}, but
the corresponding SCN-bound log-splice consumption side is
not documented in the sources surveyed. Placing an Oracle member
would require one further assumption binding the read SCN to a
consumable log position, so we do not place it here.

\paragraph{Result.} T3 on the captured scope, by
Proposition~\ref{prop:splice}. Per-table point splices at
\emph{distinct} recorded coordinates compose by
Corollary~\ref{cor:mix} to one nonempty instance at global T2, whose
conservative trajectory begins at the latest recorded splice.
Proposition~\ref{prop:sync} shows that composition alone cannot promise
a shared state witness or an earlier start. This is
Example~\ref{ex:tiers}'s composition rule as it occurs in production. One
snapshot shared across the members, its coordinate recorded, supplies
the shared state witness for T3. The modeled finite copy-before-suffix output
has final source replay under its exact scope and observation assumptions,
with consumption starting at or before the earliest recorded splice
(\S\ref{sec:merge-point}).

\subsection{Native dumps and uncertainty
intervals}\label{sec:inst-dump}

\begingroup
\RaggedRight
\paragraph{Point-splice construction.}
The maximal chunk is an entire instance captured by the engine's
own dump tooling. \texttt{mysqldump} with
\texttt{-{}-single-transaction -{}-source-data} opens a repeatable-read
transaction, records the binlog coordinate at the dump's start
behind a brief global read lock, and documents the alignment as
``in all cases, any action on logs happens at the exact moment of
the dump''~\cite{mysql_mysqldump}. \texttt{mariadb-dump} with
\texttt{-{}-single-transaction -{}-master-data} likewise records a
binlog coordinate for the consistent dump~\cite{mariadb_dump}. \texttt{pg\_dump -{}-snapshot=<snapshot\_name>} runs the
dump inside a slot-exported snapshot, splicing at the slot's
consistent point~\cite{postgres_pgdump,postgres_protocol_replication}. No closing marker exists
in any of the three. The high edge is definitionally the low edge,
and ``end of dump'' is a duration, not a coordinate. Formally all
of this is Proposition~\ref{prop:splice} again, with the dump as
the one-unit refresh. What distinguishes the dump members is the
basis of their assumptions, recorded here.
\par\endgroup

\paragraph{Native-dump assumptions.}
\begin{description}
\RaggedRight
\item[\textbf{A-cert}, \textbf{A-scope}] the native-dump member
adopts Proposition~\ref{prop:splice}'s two point bindings. The
engine view must coincide with the tool's published coordinate, and
the coordinate recorded in the dump must be that published
coordinate. Each deployment must tie these trusted assertions to
the selected tool's documented sequence. They are assumptions, not
T4 certificates.
\item[\textbf{A-engine}, \textbf{A-engine-scope}] the consistency guarantee is
engine-scoped. The documentation states that ``only InnoDB tables are dumped in a consistent
state,'' and tables of non-transactional engines dumped alongside
``may still change state''~\cite{mysql_mysqldump}. A-engine-scope
therefore excludes them from the dumped scope or the claim.
\item[\textbf{A-ddl}] the dump's consistency prohibition on
concurrent schema statements is documented and unenforced, stating that ``no
other connection should use \ldots{} \texttt{ALTER TABLE},
\texttt{CREATE TABLE}, \texttt{DROP TABLE}, \texttt{RENAME
TABLE}, \texttt{TRUNCATE TABLE}''~\cite{mysql_mysqldump}. This
model has a fixed key space and no schema transitions, so it
represents the deployment only while the prohibition holds. There
is no in-model proposition to assume, and the instance carries the
precondition in this assumption list.
\item[\textbf{A-scope-witness}] the dump invocation or manifest
enumerates the included relations and key ranges, and the dump's
recorded coordinate binds that enumeration. The coordinate can be
server-wide while the dump is partial. MySQL's emitted GTID state
covers transactions ``even those that changed suppressed parts of
the database, or other databases on the server that were not
included in a partial dump''~\cite{mysql_mysqldump}. The model
therefore uses the whole-server log but restricts the claim to the
enumerated dump scope.
\end{description}

O5 is this instance's critical obligation. The refresh
bypasses the connector's row path entirely. Rows are parsed from
dump output rather than decoded from the log, and representation
agreement between the two paths, key identity and value encoding,
is information a deployment must establish. The formal side is
satisfied when the typed refresh is constructed. The operational side,
charset, collation, and numeric edge cases agreeing between dump
rendering and log rendering, is the checkable clause that O5 makes explicit.

\paragraph{The degradation theorem.}
The remaining case concerns the \emph{untrusted} coordinate of a restored
physical backup whose asserted recovery point is not certified.
Without an exact view-coordinate binding the plan cannot exhibit the
shared-state anchor required for T3. An interval justified by recorded
bounds and containing the true view can still serve as a bracket. This
placement requires the entire restored dataset in scope to equal one
engine-consistent source state at one unknown coordinate $v$. A fuzzy
or multi-view image whose keys came from different uncoordinated views
is outside this placement unless a separate per-key mapping supplies
the contract witnesses.
Figure~\ref{fig:degradation} draws both members of the family.

\begin{figure}[t]
\centering
\includegraphics{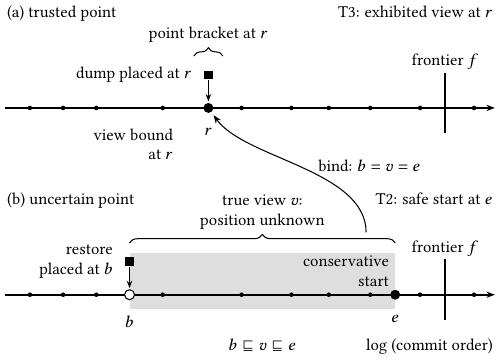}
\caption{The degradation family. A trusted coordinate gives a
point bracket with an exhibited shared-state anchor at $r$, which is
T3 (Proposition~\ref{prop:splice}). If the available recorded evidence
establishes only $b \pre v \pre e$, the widened bracket is placed at T2
under the available recorded evidence, with its canonical
trajectory guaranteed from the known coordinate $e$
(Theorem~\ref{thm:degrade}). The trusted case is the degenerate bound
$b = v = e = r$.}
\label{fig:degradation}
\Description{Two stacked panels over the same commit-order log
axis style. The trusted-point panel shows a single point bracket at
the recorded coordinate with the dump placed there and the frontier
to its right. The uncertain-point panel shows the bracket widened
between the backup start and the post-recovery coordinate, the
window shaded across its whole interior for the unknown true
view, and the restore placed at the low edge. A
label identifies the high bound as the conservative start. A
single curved arrow labeled with the collapse of the bound joins
the uncertain bracket to the trusted point.}
\end{figure}

\begin{theorem}[Degradation]\label{thm:degrade}
Consider the dump plan of Proposition~\ref{prop:splice} with
A-cert and A-scope dropped, so that no trusted binding establishes where
the view $v$ sits. Assume instead A-engine, A-engine-scope, and the
\emph{uncertainty bound} $b \pre v \pre e$, where the available recorded
evidence bounds the true view between the backup's start coordinate
$b$ and its post-recovery coordinate $e$, with $e \pre f$. The entire
restored scope equals the source at that single $v$. Widen the unit's bracket to
$(b, e)$. Then:
\begin{enumerate}
\item the widened plan satisfies the contract, and
Theorem~\ref{thm:cut} gives the cut at the frontier, reaching level T2.
Because its only high edge is $e$, Lemma~\ref{lem:latest-high} also
gives the canonical trajectory for every $g$ with
$e \pre g \pre f$.
\item the formal window-discard emission across the uncertainty
window replays to the canonical sink at $f$. This statement requires
an exact survivor listing, consumption beginning at or before $b$,
observation of the scoped log in commit order through $f$, and closure
at $e$ only after every event through $e$ has been processed. A
finite physical emission achieves this final-replay equality when it
starts from the empty state and realizes that construction with the
stated survivor and observation conditions, by
Corollary~\ref{cor:emitted-cut}. The claim is equality after replay
through $f$, not prefix equivalence of the emitted stream
(Theorem~\ref{thm:discard}).
\item the trusted instance is the degenerate member of the same
family. Proposition~\ref{prop:splice}'s plan is the widened plan
whose bound is already the point $b = v = e = r$, its coordinate
trusted through A-cert and A-scope.
\end{enumerate}
\end{theorem}

\begin{proof}
(1) O1 holds with the one widened unit. O2 holds with the unknown
view $v$ as each key's existential witness. Here $v$ lies in the
bracket by the uncertainty bound, and validity at $v$ is A-engine
with A-engine-scope. O3 and O5 are unchanged. The contract holds,
Theorem~\ref{thm:cut} applies, and the latest high is $e$. (2) is
Theorem~\ref{thm:discard} at this plan, with consumption from at
or before $b$. (3) is immediate. Assumptions A-cert and A-scope give $v = p = r$, and
Proposition~\ref{prop:splice}'s point bracket at $r$ is exactly
the member of this family with $b = v = e = r$, the bound already
a point.
\end{proof}

What degrades is T3's shared-state evidence. The latent view
$v$ still satisfies Corollary~\ref{cor:shared}'s equations
mathematically, but its operational exhibition is a meta-level fact
rather than evidence available to the capture. Under the available
recorded evidence, the interval member is therefore placed at T2 and
receives the canonical trajectory from the known coordinate $e$,
its high edge, by Lemma~\ref{lem:latest-high}. The point member is T3
because its stated preconditions bind the read view to the recorded
coordinate. Lacking that recorded exact binding adds deduplication work
across the uncertainty window. It does not remove frontier equivalence
or the post-$e$ canonical trajectory. The emitted result remains
final-only and additionally depends on the empty start, observation,
close, and exact-survivor conditions.

\begin{example}[The degradation bracket at work]\label{ex:degrade}
A backup restores key $a$ at value $4$, read at some view between
the backup's start and its recovery point. After the view, the
source updated $a$ to $6$ before the frontier. The uncertainty
window spans both events. The discipline of
Theorem~\ref{thm:degrade}(2) buffers the restored row $a \mapsto
4$, encounters the in-window update to $6$, discards the buffer
entry, and lets the event through. The replay yields $6$, the
source's frontier value, with the stale restored value never
surfacing. The untrusted coordinate cost this
reconciliation pass over $(b, e)$, and nothing else.
\end{example}

\paragraph{Result.}
A native dump with a trusted point and the stated engine,
scope, DDL, and representation conditions reaches T3. If the point
is untrusted but the available recorded evidence bounds the true view
inside an uncertainty interval, window-discard over that interval
places the plan at T2 and supplies a canonical start at the
interval's high bound. If neither
the point nor such a bound is justified, this placement does not
apply.

\FloatBarrier

\section{Generalized DBLog at a Glance}\label{sec:practitioner}

This section is for practitioners and maintainers who operate, build, or
extend DBLog, Debezium, Flink CDC, coordinate-bound capture, or native-dump
workflows. It provides an operational entry to the framework without
requiring the formal proofs. Proofs are contained in \S\ref{sec:instances}, 
while architectural boundaries, datastore assumptions, and out-of-scope 
features are cataloged in \S\ref{sec:fence}.

\subsection{What Generalized DBLog is}
\label{sec:practitioner-meaning}

\emph{Generalized DBLog} is a correctness contract for interleaving copied
database rows with an active change log. It covers the source database and
capture side of a CDC pipeline, up to the emitted output stream. It is not a
connector or an additional runtime layer. Instead, it defines what must be 
true for the output to reconstruct the source state without losing updates 
or resurrecting deleted data.

Under this contract, a connector implementation organizes its work around a 
declared scope, which is the set of row keys it intends to cover. The protocol 
divides that scope into copy units like
chunks, tables, or key ranges. For each unit, it identifies the keys the
unit owns, places the read between a low and high log position, and
reconciles changes in that interval with the copied values. The unit
closes only after the changes needed for reconciliation have been seen.
An exact-position read is the special case where the low and high
positions are the same. A whole dump can be one large unit, while parallel
capture can use many non-overlapping units with different intervals.
Figure~\ref{fig:windowing} earlier in this paper shows how the original
DBLog algorithm realizes this pattern with two watermark events around
one chunk read.

When the contract holds, replaying the combined copy-and-log output
reconstructs the exact state obtained by applying the source's committed
history up to one log position, key by key. We call that position the
observation frontier, forming a virtual cut across the whole captured scope.
The keys do not need to be read at the same time. Figure~\ref{fig:cuts} makes this concrete for a four-key scope. At $f$,
only the first unit has closed, so the reconstruction covers $k_1$ and $k_2$
while making no claim yet for the open unit (Lemma~\ref{lem:closed}). At $f'$,
both units have closed, and the four values at that cut match source-log
replay at $f'$, even though no database snapshot was taken there. At $f''$,
the reconstruction also reflects a later committed change to $k_3$. Once
all units have closed, the log-wins reconstruction continues to match
source-log replay at every later observed position.

\begin{figure*}[t]
\centering
\includegraphics{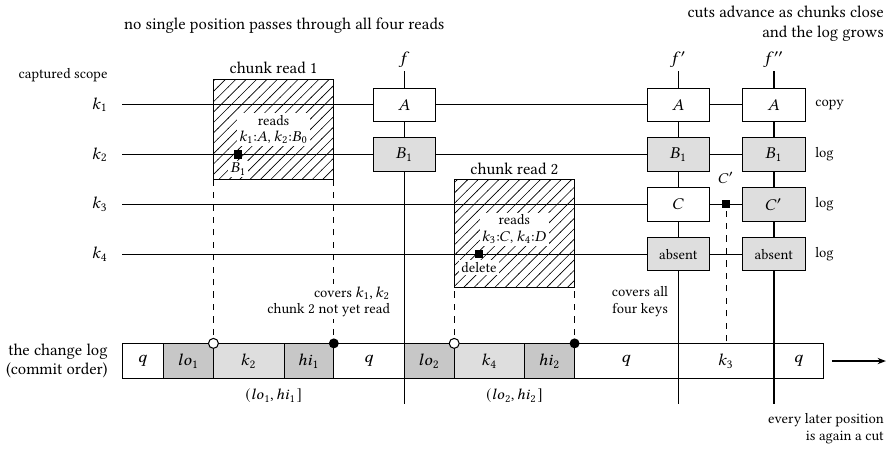}
\caption{Virtual cuts on a four-key scope. Each copied value is valid
somewhere between its unit's low and high log positions. Its exact
position is not observed and may differ by key. Reading down a cut
shows the reconstructed state for each covered key, where white
boxes indicate values retained from the row copy and gray boxes indicate
values updated or deleted by the change log. In the log, $q$ marks
committed events for keys outside the captured scope.}
\label{fig:cuts}
\Description{Four horizontal key lanes above a commit-ordered log.
Two bracketed chunk reads occur at different positions. Three vertical
cuts labeled $f$, $f'$, and $f''$ first cover two keys, then all four
keys, and finally all four keys after a later logged update. Gray
state boxes contain values supplied by the log, and white boxes contain
values supplied by the copy.}
\end{figure*}

The contract specifies correctness conditions, not operational mechanics.
Whether a pipeline relies on write-based markers (classic DBLog), sampled
GTID sets or transaction snapshots (Debezium read-only), parallel chunk
readers (Flink CDC), or
exact snapshot coordinates (MariaDB point-splice), they all fulfill the exact
same contract.
Operational choices like chunk size, bracket width, write-free
watermarking, and parallelism affect cost and throughput, but they do not
weaken the guarantee when the core conditions hold.

A practitioner can assess a protocol by asking five questions. Which
committed history is authoritative? What keys are in scope? Where does
each read belong in that history? How are racing log events reconciled
with the copy? Up to which log position is the reconstructed state
guaranteed to match the source? The following two subsections divide the answers between the
database and log on one side, and the connector implementation on the
other. Together, the two lists cover the framework's common source and
capture requirements in operational language. The formal statements
remain in Sections~\ref{sec:setting}, \ref{sec:contract},
and~\ref{sec:merge}. Individual protocol mappings may add requirements.

\subsection{What the datastore and change log must provide}
\label{sec:practitioner-source}

The framework requires several baseline guarantees from the database 
and its change log. If a database engine does not provide them natively,
the connector must bridge the gap.

\begin{enumerate}
\item \textbf{A commit-ordered log.} Every committed change appears exactly 
once in the source history, in commit order. Events from active or aborted
transactions do not appear. Only after a transaction commits does its
complete ordered block of row changes enter this history.

\item \textbf{Log positions with a precise meaning.} A scalar offset,
binlog position, LSN, or GTID set must identify exactly which prefix of
the source history it represents. The chunk reads, bracket
positions, and change-log events must all refer to that same history.

\item \textbf{A stable logical identity for each row.} The copy and log
paths must map rows to identical logical keys. A primary-key update is
modeled as deleting the old identity and inserting the new one. Tables
lacking primary keys require another unique identifier. A sortable 
key is required only when partitioning tables via ordered chunk scans.

\item \textbf{Complete row images for logged changes.} Change events must 
carry full row after-images so newer changes can replace older state. 
If the log emits only deltas with modified columns, letting the log 
override the chunk copy would lose the unchanged columns, leaving an 
empty downstream sink with incomplete rows.

\item \textbf{Committed, complete chunk reads.} The chunk scan must read 
committed data without silently skipping rows. Any row omitted by the scan 
must genuinely not exist at the time of the read. That row does not need to 
remain absent when the window closes, because a subsequent change-log insert 
will safely capture it.

\item \textbf{Sufficient log retention.} The database must retain its 
change log long enough for capture to process all events needed for reconciliation 
and continued replay. Continuous streaming connectors (like classic DBLog and Debezium) 
require the log to be retained from at or before each chunk's low watermark, while parallel 
connectors (like Flink CDC) need past log segments to remain readable 
so workers can fetch and fold changes into chunks. If required log history is 
purged prematurely, capture cannot safely reconcile the copy. Related recovery and 
failover guidance is detailed in \S\ref{sec:practitioner-lifecycle}.

\end{enumerate}

\subsection{What the connector must establish}
\label{sec:practitioner-implementation}

The connector implementation turns the source capabilities above into a complete handoff.

\begin{enumerate}
\item \textbf{A non-overlapping division of the scope.} The scope is 
the target set of keys to copy, whether a specific key range, a single 
table, or multiple tables. Every key in this scope must belong to exactly 
one unit. Units may be chunks, ranges, or full tables, but their coverage 
must neither overlap nor leave gaps. The scope also includes keys that 
are absent initially and may be inserted while capture runs.

\item \textbf{Bracketed reads for every chunk.} Each chunk has a low and
high log position. Each row key in the chunk's range must reflect committed
source state at some position inside that bracket. Different keys may match
different coordinates. A single read is sufficient, but an implementation 
can also issue multiple reads within the bracket.

\item \textbf{Matching data types and operations.} Because chunk reads 
and change-log events arrive through different database interfaces, the 
connector must deserialize both into a canonical format and matching types.
Operation meaning (inserts, updates, and deletes) must also be preserved 
during reconciliation rather than mapped into generic appends.

\item \textbf{The log must win over copied data.} Log events must always 
take precedence over older chunk reads. The reconciliation logic must ensure 
that stale copied rows never overwrite newer updates or resurrect deleted rows.

\item \textbf{No chunk closes before its window is fully processed.}
A chunk cannot emit its rows until all change events that raced the read have 
been reconciled. In streaming connectors (classic DBLog and Debezium), the 
connector tails the log starting at or before the low watermark, discarding 
buffered chunk rows as newer log events arrive, and releasing surviving 
rows only after consuming past the high watermark. In parallel connectors 
(Flink CDC), workers fetch the log segment between the low and high offsets
from the server's retained binlog to fold changes into the chunk. The live reader starts 
at or before every chunk's high offset. It emits an event only when its 
position is strictly beyond the high offset of the event's owning chunk.

\item \textbf{Final state equivalence of the output stream.} The output  
stream emitted by the connector does not need to match the historical 
change log event-for-event. However, replaying that emitted stream from an 
empty state must produce the exact same final state for every key as 
replaying the source log. This equivalence must hold across the entire 
handoff.
\end{enumerate}

These conditions establish the copy-to-log guarantee described above.
A production system that can restart or fail over has an additional lifecycle
responsibility. It must recover the cursor, buffer, gate, and close state
consistently, or fail closed and rebuild the incomplete unit. That
guidance is discussed in \S\ref{sec:practitioner-lifecycle}.

\subsection{How to read the observation frontier}
\label{sec:practitioner-frontier}

The observation frontier is a single log position for the entire captured 
scope. The theorem guarantees that replaying the copy and log up to this 
frontier matches the source change log, key by key. However, because a 
single committed transaction may modify multiple rows, a log coordinate 
can land in the middle of that transaction's row events.

For example, suppose keys $a$ and $b$ both start at $0$, and a transaction 
commits changes setting both to $1$. Its committed block of events is 
$[a \mapsto 1,\ b \mapsto 1]$. If replay stops after the first event, the 
reconstructed state is $a=1,\ b=0$. This state exposes no uncommitted data, 
but it is an intermediate step that never existed as an observable database 
state.

At any frontier, Generalized DBLog guarantees exact row-event replay. If 
the frontier also lands on a transaction boundary (known as \emph{transaction 
closure}), the reconstructed state is also a consistent transactional 
snapshot. In short, reading committed data ensures clean input, but 
stopping at a transaction boundary is what guarantees a transactional snapshot. 
Protocols that select frontiers from transaction-complete coordinates 
(such as MySQL GTIDs) satisfy this condition naturally.

A connector that requires a transactional snapshot must ensure its 
frontier aligns with a transaction boundary. However, this does not guarantee 
atomic downstream application. A sink database that applies events individually 
may still expose intermediate states to its own readers. Downstream sink 
behavior remains outside this result.

\begin{table*}[t]
\caption{Generalized DBLog protocol mappings. The detailed conditions
appear in \S\ref{sec:instances}.}
\label{tab:instances}
\footnotesize
\setlength{\tabcolsep}{3.5pt}
\renewcommand{\arraystretch}{1.08}
\begin{tabular}{@{}>{\RaggedRight\hyphenpenalty=10000\exhyphenpenalty=10000\arraybackslash}p{3.6cm}
                    >{\RaggedRight\hyphenpenalty=10000\exhyphenpenalty=10000\arraybackslash}p{6.05cm}
                    >{\RaggedRight\hyphenpenalty=10000\exhyphenpenalty=10000\arraybackslash}p{7.1cm}@{}}
\toprule
Protocol shape & Bracket and scope & Reconciliation \\
\midrule
\textbf{Classic DBLog} &
Chunk between two captured marker events. &
Window-discard. \\
\addlinespace
\textbf{Debezium signal-table}\newline(\texttt{insert\_insert},\newline\texttt{insert\_delete}) &
Chunk bracketed by an opening signal-table event and either a closing
insert or deletion of the opening row. &
Window-discard, where a matching in-window event discards the buffered copy. \\
\addlinespace
\textbf{Debezium MySQL}\newline\mbox{read-only} &
Chunk between exact complete executed-GTID prefix sets. &
Window-discard by membership in high minus low. \\
\addlinespace
\textbf{Debezium MariaDB}\newline\mbox{read-only} &
An empty difference gives a conditional point placement. A nonempty difference
is placed here only for an empty buffer. Populated nonempty cases remain open. &
Point splice at a confirmed quiet point. Otherwise window-discard with
no copied survivors. \\
\addlinespace
\textbf{Debezium PostgreSQL}\newline\mbox{read-only} &
Transaction snapshots taken before and after the chunk read. &
Window-discard over the sampled interval or a wider window determined by
transaction-ID comparison. \\
\addlinespace
\textbf{Flink CDC MySQL}\newline parallel chunks &
Final logical chunks with low and high binlog offsets. &
Range-merge each complete window, then gate the live suffix per split. \\
\addlinespace
\textbf{Coordinate-bound}\newline read or trusted dump &
One scope, or multiple scopes sharing one exact read coordinate. &
Copy-before-suffix point splice. \\
\addlinespace
\textbf{Bounded-uncertainty}\newline restore &
The entire restored scope equals one source state at one unknown point
inside justified bounds. &
Window-discard over the whole interval. \\
\bottomrule
\end{tabular}
\end{table*}

\subsection{How the protocol families fit}
\label{sec:practitioner-table}

Table~\ref{tab:instances} shows how each protocol family establishes a 
bracket and reconciles copied rows with the change log. The framework's 
correctness guarantees hold for a protocol when it satisfies both the baseline 
requirements from \S\S\ref{sec:practitioner-source}-\ref{sec:practitioner-implementation} 
and the protocol-specific conditions described below. The table classifies 
these connector architectures rather than certifying individual releases.

The protocols rely on three reconciliation rules:
\begin{enumerate}
\item \textbf{Window-discard.} The connector buffers copied rows and
forwards log events. An in-window event removes the buffered row with the
same key. Only the surviving copied rows are emitted when the window closes.
\item \textbf{Range-merge.} The connector folds in-window log events directly 
into the copied rows, emits the reconciled chunk when the window closes, and streams 
subsequent log changes individually.
\item \textbf{Point splice.} The zero-width window case ($lo = hi$). The copy is 
placed at its exact log coordinate, and live log streaming resumes immediately from 
that point. When capturing multiple tables via point splices, log streaming must begin 
at or before the earliest recorded coordinate.
\end{enumerate}

Under our baseline conditions, all three strategies produce the exact same final 
state as replaying the source log. Replaying the output stream matches the source 
database through the latest chunk's high position and continues to match as 
subsequent log events arrive.

\subsection{Conditions for the named protocols}

Under the stated datastore and implementation conditions, the analyzed DBLog,
Debezium, and Flink CDC protocol designs fit Generalized DBLog. These
mappings identify the exact properties that a connector release and engine
configuration must establish. The baseline
datastore and connector requirements apply to every row in Table~\ref{tab:instances}. 
The paragraphs below outline the specific handoff mechanisms and operational 
assumptions for each protocol.

\paragraph{Classic DBLog.}
The watermark table must be writable, both markers must be recognizable in 
the change log, and the chunk scan must cover its key range completely. 
The connector must consume past the high marker before emitting surviving chunk rows. 
A chunk read that starts after the low marker write has completed must reflect
every transaction logged before that marker (W6). On a PostgreSQL primary,
which does not guarantee this, the connector must apply the snapshot
comparison of \S\ref{sec:inst-classic}.
Theorem~\ref{thm:classic} proves that this original watermark protocol 
satisfies the generalized contract.

\paragraph{Debezium signal-table modes.}
In the default signaling mode (\texttt{insert\_insert}), opening and closing records in a 
signaling table act as the low and high watermarks. The opening record 
must precede the chunk read, the closing record must follow it. The 
connector must not emit surviving rows until all change events within 
that window have been processed. Copied rows and change-log events must 
agree on key and value representations. Under \texttt{insert\_delete}, the opening 
insert marks the low watermark and its subsequent deletion marks the 
high watermark. Debezium converts that deletion into an internal close 
signal, provided the delete event's before-image retains the primary 
key~\cite{debezium_signalling,debezium_signal_delete_source,debezium_mysql_connector}.
The mode and these conditions are the same for the PostgreSQL
connector~\cite{debezium_postgresql_connector}. On a PostgreSQL primary, the connector
must also apply the snapshot comparison described for classic DBLog.

\paragraph{Debezium MySQL}
Debezium's read-only MySQL mode uses \texttt{Executed\_Gtid\_Set} samples as 
chunk watermarks. Because MySQL documents non-atomic GTID externalization, 
each sampled GTID set must map to an exact, transaction-complete prefix 
of the change log~\cite{debezium_mysql_connector,mysql_gtid_lifecycle}. 
In addition, the chunk read must cover its assigned key range, and surviving rows 
must not be emitted until the window reaches the closing GTID sample.

\begingroup
\RaggedRight
\paragraph{Debezium MariaDB}
MariaDB's \texttt{GTID\_\allowbreak{}BINLOG\_\allowbreak{}POS} records one frontier per domain,
not MySQL's complete executed set. Debezium 3.6.1.Final samples this
vector and requests a reread when the high-minus-low GTID difference
is nonempty. The helper performs the reread only when its snapshot,
window-state, and buffer guards hold
~\cite{debezium_mariadb_connector,debezium_mariadb_readonly_source,debezium_mysql_readonly_source,mariadb_gtid}.
Under the same-history and monotone-prefix assumptions, the 
conditional bridge accepts an empty difference as a confirmed quiet
point only when the complete vector maps to the exact consumed prefix
and covers all relevant domains.
Under those assumptions, an empty difference gives a point splice for
either a populated or empty chunk. With a nonempty difference, an
empty buffer contributes no copied rows, and under the stated read,
observation, and window-discard rules, any concurrent insert remains
on the log path. These two analyzed cases use different proof
structures but yield the same base guarantee (\S\ref{sec:inst-readonly}).
These results do not cover a populated buffer with a nonempty GTID
difference. Neither analyzed case uses the
MySQL set-difference theorem or MySQL replica variables. Independently
advancing domains still require an explicit same-history,
prefix-preserving mapping.

\par\endgroup

\paragraph{Debezium PostgreSQL}
Debezium's read-only PostgreSQL mode requires PostgreSQL 13 or later. It
calls \texttt{pg\_\allowbreak{}current\_\allowbreak{}snapshot()} before
and after each chunk read, then compares streamed transaction IDs with
the sampled \texttt{xmin} and \texttt{xmax} boundaries to decide when the window opens and closes
~\cite{debezium_ddd8,debezium_postgres_readonly_source}.

Each sample must represent a complete prefix of the consumed change
log. When the low sample takes effect, it must include every transaction
the connector has already processed. The streamed transaction IDs and
sampled bounds must also use the same ID space. A PostgreSQL primary
guarantees neither of the first two conditions (\S\ref{sec:inst-readonly}). Under these
conditions, the connector's window covers the sampled interval and may
extend beyond it. The mode therefore yields the same base guarantee as
the MySQL mode.

Both PostgreSQL modes require update events to carry complete row
images. For tables with large column values, this requires
\texttt{REPLICA IDENTITY FULL} or an equivalent
normalization~\cite{debezium_postgresql_connector}.

For all three read-only variants, the connector must record the low coordinate
before starting the read and the high coordinate after the read completes.
Surviving chunk rows may only be released after processing all row events
belonging to the window's final transaction. Out-of-scope events and
server heartbeats needed to advance log coordinates must remain observable.

\begin{table*}[t]
\caption{Operational lifecycle guidance outside the copy-to-log proof.}
\label{tab:lifecycle}
\footnotesize
\setlength{\tabcolsep}{2.6pt}
\renewcommand{\arraystretch}{1.08}
\begin{tabular}{@{}>{\RaggedRight\hyphenpenalty=10000\exhyphenpenalty=10000\arraybackslash}p{2.2cm}
                    >{\RaggedRight\arraybackslash}p{3.15cm}
                    >{\RaggedRight\arraybackslash}p{4.1cm}
                    >{\RaggedRight\arraybackslash}p{3.45cm}
                    >{\RaggedRight\arraybackslash}p{4.05cm}@{}}
\toprule
Discipline & Source history and\newline recovery floor & Recoverable unit state and close condition & Restart or rebuild action & Safe purge and fail-closed rule \\
\midrule
\textbf{Window-discard} &
History from every open unit's low edge and from the continuing reader. &
Scope, low and high, survivor buffer, processed cursor, and whether every
event through the high has closed the unit. &
Resume from one consistent checkpoint, or discard the incomplete unit
and reread it under a new bracket. &
Purge only below every durable unit and reader floor. If required
history is gone, do not release survivors. Rebuild the unit. \\
\addlinespace
\textbf{Range-merge} &
Complete low-to-high backfill for each unit and a live floor at or
before every high. &
Logical ownership, offsets, merged result, processed cursor, per-key
gate, and completion of the high-edge merge. &
Resume buffer and gate together, or discard the incomplete merged unit
and reconstruct it from retained history. &
Purge only when the merged result and gate durably represent the
consumed range. If either cannot be recovered, fail closed and rebuild. \\
\addlinespace
\textbf{Point-splice} &
History from the splice coordinate for the continuing suffix. &
Coordinate-bound scope, copy completion, suffix cursor, and handoff
state. &
Resume only while the same read view or completed copy remains valid.
Otherwise obtain a new bound copy and splice point. &
Purge below the splice only after the copy and suffix floor are durable.
If the view or suffix is unavailable, do not claim the splice. \\
\bottomrule
\end{tabular}
\end{table*}

\paragraph{Flink CDC MySQL}
Flink CDC partitions tables into chunk splits and uses range-merge reconciliation. 
The connector folds change-log events directly into each chunk split and emits 
the merged rows at the split's high binlog offset. A live reader then streams 
subsequent change events. To prevent duplicate emissions, Flink uses a per-split 
gate that withholds live stream events whose positions fall at or before that split's 
high offset. Chunk splits can be processed in any order, provided the split key is immutable. 
Updating a split column mid-capture can route rows across splits unpredictably.

Replay equivalence requires Flink's changelog to preserve row operation metadata 
(RowKind for inserts, updates, and deletes). Append-only formats that discard operation 
types before reconciliation do not satisfy the contract. Backfill-skipping modes and tables 
lacking primary keys are not covered. Pipeline restarts, dynamic rescaling, failover, and 
sink-side exactly-once processing remain separate operational concerns.

\paragraph{Native dumps and physical restores.}
A point splice models native dump tools (such as \texttt{mysqldump}, \texttt{pg\_dump}, 
and \texttt{mariadb-dump}) that bind an engine-consistent view to an exact recorded log 
coordinate. Because the window width is zero, no reconciliation is needed. Log consumption 
begins at or before that coordinate. In contrast, a bounded restore models physical backups 
and storage snapshots where the restored state is engine-consistent at an unknown coordinate 
within recorded bounds. The connector streams the log from the backup's start, using 
window-discard to drop any restored rows modified before recovery completes. Replay is exact, 
and tracking the live source is guaranteed from the high bound onward.

\subsection{Retention, restart, and failover guidance}
\label{sec:practitioner-lifecycle}

Table~\ref{tab:lifecycle} gives operational guidance for preserving the
source and emission conditions across crashes, restarts, and database failovers.
Its retry and rebuild actions must account for prior output and every key
affected by lost history, as required in \S\ref{sec:practitioner-source}.
These are conditions for a recovery design, not a proved lifecycle protocol.

Safe crash recovery requires coordinating database log retention with 
connector checkpoints. The database must retain change logs back to the 
oldest position required by any open chunk or active reader. Purging 
logs ahead of that position prevents reconciliation, forcing the connector 
to fail closed and reread the affected units. When persisting connector 
progress, saved log offsets must match the state of in-memory chunk 
buffers and merge gates, using a transaction, write-ahead journal, or 
idempotent replay to survive worker restarts without drift.

Following a database failover, replication coordinates must map to the same continuous 
history or a verified, prefix-preserving continuation on the new primary. If coordinate 
continuity cannot be guaranteed, the connector must fail closed and rebuild the affected 
chunks rather than risk data corruption across divergent logs.

\FloatBarrier

\section{Making the 2020 Requirements Explicit}\label{sec:audit}

The 2020 DBLog paper listed two database requirements for capture, namely  
commit-ordered change events and non-stale reads~\cite{andreakis2020dblog}. 
While our framework validates the original algorithm under its stated assumptions,
it also clarifies that
these requirements were underspecified. On the read side, value
freshness alone is not enough. The chunk scan must completely cover its assigned 
key range without silently skipping rows. If an unchanged row is missed by the 
scan, no change-log event will arrive to repair it.
On the write side, watermarking implicitly requires table provisioning, write 
permissions, log coverage, recognizable marker events, an ordered log, and
visibility order.
The generalized contract formalizes these properties as explicit read 
obligations (O2 and O3, \S\ref{sec:contract-obligations}) and watermark 
preconditions (W1 through W6, \S\ref{sec:inst-classic}), examined in
Sections~\ref{sec:audit-read} and~\ref{sec:audit-write} below.

\subsection{Read completeness}\label{sec:audit-read}

The 2020 paper states its chunk-read requirement twice, at different 
strengths. The strong form requires that ``the selection executes on a specific 
position of the transaction log.'' That means a single coordinate, which
satisfies O2 with every key witnessed uniformly and is
sufficient by Theorem~\ref{thm:classic}. The text then generalizes, requiring 
that ``the chunk selection sees the changes that are committed before its 
execution,'' names this capability \emph{non-stale reads}, and includes only 
this weaker condition in its Database Support section as the portability 
criterion~\cite{andreakis2020dblog}. The two statements are not equivalent, 
and this difference is precisely what Theorem~\ref{thm:cut} formalizes.

Per-row freshness already handles values that change during the scan. If a key 
changes inside the window, its log event removes or replaces the buffered row. 
If a key does not change, a value reflecting the source state at any coordinate 
in the bracket remains valid. This is why the theorem does not require all keys 
to share one read coordinate.

The phrase ``committed before its execution'' can be read in two ways.
The plain reading requires the chunk read to reflect
every transaction whose commit had completed when the read started. The log-order
reading requires it to reflect every transaction logged before the low
watermark. W6 is the log-order reading
(\S\ref{sec:inst-classic}). Under MySQL's ordered commits, a read that
satisfies the plain reading and starts after the low watermark write has
completed also satisfies the log-order
reading~\cite{mysql_binlog_options,mysql_consistent_read}. On a PostgreSQL
primary this implication does not hold. A chunk read can then miss a
transaction logged before the low watermark, and O2 can fail for the keys
that transaction changed.

Scan completeness addresses a different requirement. Freshness constrains the 
values of rows that return, but says nothing about which rows return. Suppose a 
row remains present throughout the window and the scan silently skips it. Because 
the row never changed, the change log contains no event to supply it, and the connector 
advances to subsequent chunks without ever rescanning that range. The sink remains 
missing that row until some later write happens to update it. Even if the total refresh represents this omission as
absence to satisfy O3 totality, O2 bracket validity still fails because the
row remained present throughout the entire bracket. Together, O2 and O3
require every window-invariant present key to be returned with its valid value
(\S\ref{sec:contract-obligations}).

The vendor documentation shows why this distinction matters. On MySQL and 
PostgreSQL, a read-committed select executes against a single statement 
snapshot~\cite{mysql_consistent_read,postgres_transaction_iso}. The 2020 
deployments therefore satisfied scan completeness automatically. Under
SQL Server's default locking read committed, shared row locks are released as 
the scan advances. Microsoft's locking guide documents that a concurrent key 
update can move a row so that it appears twice or is skipped 
entirely~\cite{mssql_locking_guide}.

That documented case updates the row during the scan, meaning its log event 
would repair the omission under Case~A of Theorem~\ref{thm:cut}. While it is not a direct 
example of an unrepaired missing row, it establishes that lock-based scans can 
silently skip rows as concurrent updates move them across pages. While we do not 
claim a live reproduction of data loss, the architectural conclusion remains 
clear: ``non-stale reads'' alone does not guarantee the scan completeness that 
DBLog relies on, whereas an engine statement snapshot guarantees it by design.

The complete condition is the contract's combination of O2 with O3, with a 
statement snapshot at one in-bracket coordinate serving as the canonical
sufficient condition and the strong
2020 formulation as its uniform special case. Read against this, the two requirements
listed in the 2020 paper were not sufficient for PostgreSQL, an engine that
paper lists as supported. A PostgreSQL primary can satisfy both, with
non-stale reads in the plain reading, and still violate W6.
An isolation-level label does not by itself establish complete per-key bracket
validity.

\subsection{Watermark preconditions}\label{sec:audit-write}
The same Database Support section in the 2020 paper lists only commit-ordered 
change events and non-stale reads. The mechanism description, however, creates a 
dedicated watermark table and brackets every chunk with two updates whose events 
must appear in the change log~\cite{andreakis2020dblog}. Writability is therefore 
required by the mechanism but absent from the compatibility checklist. This 
matters for the non-relational stores the original text intended to cover,
and also for read-only replicas. They can satisfy both listed conditions
while being unable to
run the watermark path.

Section~\ref{sec:inst-classic} makes these requirements explicit. Conditions W1 
through W6 state what the watermark mechanism actually requires, from table
provisioning and write access to capture coverage, recognizability,
same-instance ordering coupling, and visibility order. None of them is new. W1
through W5 are drawn directly from the 2020 paper's own mechanism description,
and W6 is the log-order reading of non-stale reads
(\S\ref{sec:audit-read}). What is new is recognizing them
as preconditions of one specific bracket construction rather than universal 
requirements of change-data-capture. Writability is a property of the chosen 
bracket mechanism, not of the datastore. The write-free instance of 
\S\ref{sec:inst-readonly} satisfies the same contract obligations with no writes 
whatsoever, proving constructively that the 2020 watermark writes were simply one 
way to build brackets rather than an inherent requirement for correctness.

Both clarifications lead to the same conclusion. The original watermark 
mechanism remains sound, and the generalized contract cleanly separates that 
mechanism from the conditions needed to prove it correct. 
Section~\ref{sec:instances} then applies those same conditions to mechanisms 
that construct brackets without writes.

\section{Mechanization}\label{sec:mech}

The normalized mathematical content of every numbered definition is
represented, and every numbered result from the cut theorem through the
five placements and the degradation pair is proved in a machine-checked
Isabelle/HOL development comprising about 5{,}000 lines of definitions and
structured proofs, checked with Isabelle2025-2
over the standard library. Coordinate fidelity, O4, raw representation
agreement under O5, and connector-release conformance are treated as 
instance evidence outside the prover. The development grounds the
title's ``Verified'' for the normalized theorem chain. The body's mathematics
stands alone. Each printed proof can be read against its mechanized
counterpart or without it. The correspondence is organized by the
paper's own definition and theorem numbers.

The three verification layers have different scopes:
\begin{description}
\RaggedRight
\item[Isabelle/HOL.] The full development covers the normalized
mathematical form of every numbered definition and result, including
the five placements and degradation. It contains about 5{,}000 lines and was checked by a clean session
build. This is the paper's general machine-checked proof.
\item[Lean~4.] The independent re-proof covers the contract core,
cut theorem, the first three corollaries, window invariance,
window-discard equivalence and monotonicity, and the
no-shared-read-point witness over the core library. It is not a port of
every placement or merge.
\item[TLA\textsuperscript{+}/TLC.] Five models cover the classic,
read-only, parallel, dump/degradation, and heterogeneous-composition
protocols. Across 25 bounded runs they explore up to three
keys, four events, and two chunks. The largest search has about 74
million distinct states. This is bounded evidence, not proof.
\end{description}

The classic instance imports the published formalization of the original
algorithm through a session refactoring. Shared source and replay definitions
are moved into a base session, and two identical virtual-cut predicates
are consolidated. The wellformedness clauses and mathematical results
are unchanged from the published development
archived at~\cite{andreakis2026virtualcuts_artifact}.
Theorem~\ref{thm:classic} adopts those wellformedness clauses directly,
and its source-state conclusion agrees with that development
(\S\ref{sec:inst-classic}).

The development is held to the submission standards of Isabelle's
Archive of Formal Proofs. It is axiom-free, with no statement left
unproved and no proof step appealing to an external oracle.
Definitions are conservative constructions, and the proofs are
structured, human-readable Isar under a typeset proof document.
The Isabelle session builds the complete chain, including the
imported corpus, from a clean environment with permissive mode
disabled. A clean build therefore checks every stated result rather
than reusing a cached proof image.

The theorems are exercised as well as proved. Every countermodel in the
numbered results exists in the development as a constructed instance.
The no-shared-read-point instance of Proposition~\ref{prop:torn}, the
no-early-onset composite of Proposition~\ref{prop:sync}, and the
window-narrowed variant of Proposition~\ref{prop:withhold} each carry
their failure proofs there, establishing no single shared read coordinate, no early
onset, and a wrong replay, respectively. The degradation bracket of
Example~\ref{ex:degrade} is likewise mechanized as a concrete instance, 
and the positive theorems are instantiated at concrete plans in the same 
style rather than left abstract. The PostgreSQL passage of
\S\ref{sec:inst-classic}, including the snapshot comparison, argues in prose
from PostgreSQL's commit procedure and Theorems~\ref{thm:cut}
and~\ref{thm:discard}, and it is not mechanized.

\paragraph{Where the effort went.}
The proof effort was not spread evenly across the development. The
cut theorem itself was the cheap part. Its per-key argument is
short in this paper and not much longer in the mechanization. Most
of the effort went into the layers that keep it short. We kept the
cut theorem free of observability, per-unit finiteness, and width
assumptions, so those concerns had to be proved at the stream layer,
where they apply. The two buffering disciplines account for
roughly a third of the development's lines, more than the contract
and the cut theorem together. Most of that is in per-key emission-shape
lemmas, duplicate-free survivor enumerations, and position tags.
The countermodels were the other large expense. The theory that
only constructs witnesses and their failure proofs is as long as
the theory that holds the cut theorem. The remaining difficulty
was discipline rather than proof. The classic formalization is
imported as a separate session, so we matched its conventions exactly, with
coordinate ties resolved by position and prefixes interpreted
at-or-before, instead of adjusting them to fit. We consider the
distribution itself informative. The mathematics was hardest
where production systems differ, at the emissions and their
equivalences, and not at the cut.

\paragraph{Model checking as a second check.}
Five TLA\textsuperscript{+} specifications model the classic
watermarked loop, read-only executed sets, parallel range-merge,
dump degradation, and a heterogeneous two-unit composite. They use
the same half-open windows, per-key validity, and emission rules as
the paper. Across 25 bounded runs, TLC checked frontier
replay, per-key monotonicity for the chunked emissions, executed-set
window membership, the trusted-splice and safe latest-high trajectories,
degradation, and
mixed composition. These checks are bounded. The general result comes
from the machine-checked proof above.

Mutated runs recover the important failures. A forced shared read
instant rejects Proposition~\ref{prop:torn}'s per-key-witness
construction. A
window-only withholding rule produces Proposition~\ref{prop:withhold}'s
wrong replay. Distinct splice coordinates refute an early common onset
as in Proposition~\ref{prop:sync}. A false trusted recovery point
produces a stale splice, while a claim beginning inside the widened
uncertainty bracket fails as expected. These trajectory probes refute
only the stronger earlier-onset claims. They do not refute the safe
canonical trajectory beginning at the latest high edge. All contract
runs append transactions atomically and check that no uncommitted event
is visible. In one committed-only mutation, the current executed set is
mapped to an interior prefix within an already
committed transaction block. This breaks Theorem~\ref{thm:oracle}'s
membership test without introducing an in-progress event. In a separate,
explicitly out-of-contract probe, one transaction is artificially exposed
in two pieces. The committed-only invariant fails immediately and the
membership test can fail after a watermark is sampled in between. The
boundary-valid window-discard replay survives both mutations on the small
bounded searches, but this diagnostic robustness does not admit
uncommitted input into the theory. TLC establishes bounded counterexamples 
when these assumptions are omitted, rather than a general theorem about replay. The models
carry latent read positions as state, so TLC can evaluate the resulting
equations. They do not model whether a running capture knows or can
operationally exhibit such a position.

\paragraph{A second prover.}
The contract core, Theorem~\ref{thm:cut},
Corollaries~\ref{cor:shared}, \ref{cor:cont}, and~\ref{cor:mix},
and the no-shared-read-point instance of Proposition~\ref{prop:torn},
are also re-proved in Lean~4 over its core library alone, so that a
second, independent proof kernel checks the central results.

\paragraph{Artifact availability.}
The accompanying artifact contains the Isabelle/HOL development and its
imported dependencies, the Lean~4 formalization, and the five
TLA\textsuperscript{+} models, together with build instructions and
verification records~\cite{framework_artifact}.
The artifact's Zenodo DOI is
\href{https://doi.org/10.5281/zenodo.22643866}{\nolinkurl{10.5281/zenodo.22643866}}.

\section{Related Work}\label{sec:related}

ARIES provides a useful comparison from database recovery~\cite{mohan1992aries}. 
Its fuzzy 
checkpoints record transaction and dirty-page metadata while execution 
continues, without requiring dirty pages to be flushed. Restart uses 
page LSNs to guide redo and performs transaction undo where required. 
Generalized DBLog also reconciles state with ordered history, but works 
with committed row states and full post-images. O2 requires bracket-local 
validity without recording a version for each copied component. The 
comparison is therefore between reconciliation problems, not identical 
checkpoint or replay mechanisms.

Online schema-change tooling runs the same interleave inside a
single database. gh-ost builds a ghost table, copies the original
into it incrementally, and applies concurrent changes captured
from the binary log onto the ghost before a final
cut-over~\cite{ghost_docs}. pt-online-schema-change copies in
chunks while triggers on the original forward concurrent
modifications into the copy~\cite{ptosc_docs}. The anatomy of \S\ref{sec:contract} is recognizably present, a copy, racing changes, reconciliation, and a cut-over, with the log as the reconciliation channel in one tool and triggers in the other. Both
tools are described here, and neither is claimed. Whether either
reconciliation fulfills the equivalence obligation of
\S\ref{sec:contract-replay} is left open, and the trigger channel
in particular is not a commit-ordered log, so even admissibility
(\S\ref{sec:setting-coords}) would need its own argument.

The splice shape recurs outside databases. Kubernetes' list-then-watch
protocol begins by listing a collection, takes the returned
\texttt{resourceVersion} as its coordinate, and watches for changes
after that version. When the server no longer covers that version it
answers with 410 Gone, and the client discards its state and lists
again~\cite{k8s_api_concepts}. The client lists, records a coordinate, resumes from the coordinate, and re-bootstraps when the log no longer reaches it. This is the point splice of \S\ref{sec:merge-point} in an API-server setting, and the re-list rule handles a failure this paper treats as an observability failure (S-OBS). The shape, not
membership, is what the citation records. Whether
\texttt{resourceVersion} admits a faithful prefix reading is not
examined here. Within databases, PostgreSQL's logical replication
bootstraps each table with its own tablesync worker before
handoff to the shared apply
stream~\cite{postgres_logical_replication}.
Section~\ref{sec:tiers-example} names the resulting composite after
exactly this shape, and Proposition~\ref{prop:sync} gives its
composition rule. Vitess VReplication copies key ranges and gates the
event stream per copied range~\cite{vitess_vreplication}. Its
discipline is the key-range entry of \S\ref{sec:merge-open},
stated there as an open obligation.

In method, this paper belongs to a verification literature that rarely treats this subject. Machine-checked systems artifacts run from
crash-safe storage, where specifications carry explicit crash
conditions~\cite{chen2015fscq}, to distributed-system implementations
proved correct against refinement
specifications~\cite{hawblitzel2015ironfleet}. Industrial
practice, separately, checks designs by bounded exploration of
formal specifications~\cite{newcombe2015amazon}. This paper does
not verify an implementation. It verifies a contract and proves
conditional placements for documented protocol designs. In this
respect it is closer to specification frameworks that characterize
a database guarantee by what clients can observe~\cite{crooks2017seeing}
than to verified implementations.

The nearest works are this paper's own predecessors. The classic
algorithm's formalization~\cite{andreakis2026virtualcuts} proves
the certified virtual-cut object this framework incorporates at T4,
and the state in Theorem~\ref{thm:classic} is precisely its
replayed row-event state (\S\ref{sec:inst-classic}). The delivery-side
companion~\cite{andreakis2026dualwrite} studies recovery and sink acceptance
after capture, outside the scope of this paper (\S\ref{sec:fence}). The 2020
paper~\cite{andreakis2020dblog} remains the family's operational
source, and \S\ref{sec:audit} completes its requirements in the
framework's vocabulary rather than revisiting its design.

\section{Scope and Limitations}\label{sec:fence}

The framework defines the capture contract up to the emitted event stream, 
proving that replaying this stream from an empty initial state reconstructs the source database. 
While the formal theorems assume an empty state, the emitted stream is also expected to work 
when updating an already populated sink. Copied rows and logged changes update the corresponding 
rows in the sink. However, a row that no longer exists at the source can remain in an already 
populated sink if no delete is sent for it. In practice, clearing the target scope beforehand or 
emitting explicit tombstones resolves this difference. 

Extending the formal proofs to verify such populated-sink repairs, 
alongside transport, duplicate delivery, and exactly-once application, remains outside the current 
framework. Delivery is treated separately in~\cite{andreakis2026dualwrite}. Schema and DDL 
evolution are excluded by the stable logical identity space. The dump requirements in 
\S\ref{sec:inst-dump} reflect this limitation as a documented restriction on concurrent DDL.

Checkpoint, restart, rescaling, and failover protocols are not proved.
The operational table in \S\ref{sec:practitioner-lifecycle} is guidance
about preserving protocol invariants, not a verified lifecycle result. A 
resumed run obtains the theorem's guarantee only when its recovered plan
and output again satisfy the same source and emission conditions.

The analysis is based on the 3.6 documentation for both projects and the
source code of Debezium 3.6.1.Final and Flink CDC 3.6.0. These results evaluate protocol
designs rather than certifying individual software releases. A connector realizes these 
guarantees whenever its runtime execution fulfills the modeled bracket and emission rules 
under the stated preconditions.

One common requirement across all protocols concerns scope discovery. Enumerating 
tables and key ranges from the database catalog is itself a read that must correspond to a valid log 
coordinate (\S\ref{sec:setting-fence}). For native dumps, this enumeration is bound directly to the tool's 
recorded server-wide snapshot position.

The paragraphs below catalog the specific operational assumptions and open evidence questions 
for each analyzed protocol.

\paragraph{Debezium signal-table modes.}
The conditional bridge covers the documented \texttt{insert\_insert}
and \texttt{insert\_delete} strategies. It requires both window-edge
events to enter the consumed history in the correct order around the
chunk read, together with complete chunk read results (A-read). 
The connector must discard buffered rows matching in-window change events 
(A-discard), withhold surviving chunk rows until the closing signal is consumed 
from the log, and observe all intermediate log events without gaps (S-OBS). 
The signaling collection must also satisfy the watermark
preconditions (W1--W6, \S\ref{sec:inst-classic}). On a PostgreSQL primary,
the snapshot comparison described there takes the place of W6. Enabling the
source signal channel does not itself prove that its collection is
captured or its records recognizable. Under \texttt{insert\_delete},
the deletion supplies the close only when its before-image retains the
signal row's identifier. No Debezium release is verified by this bridge
~\cite{debezium_signalling,debezium_mysql_connector,debezium_signal_delete_source}.

\paragraph{Debezium MySQL}
The direct set bridge requires every accepted
\texttt{Executed\_Gtid\_Set} value to map to an exact gap-free prefix
of complete transactions in the same history used by the chunk read
and consumed events (A-linear, \S\ref{sec:inst-readonly}). Because MySQL
explicitly documents non-atomic GTID externalization
~\cite{mysql_gtid_lifecycle}, the framework does not infer the
needed semantic prefix property solely from when the status variable
becomes visible. Complete chunk reads (A-read below), raw-to-projected
GTID alignment, stable identity (A-scope), full row images (S-IMG), observation of
state-advancing filtered events (S-OBS), and window-discard timing (A-discard)
remain explicit conditions.

\paragraph{Debezium MariaDB}
Unlike MySQL's executed GTID set, MariaDB's \texttt{GTID\_\allowbreak{}BINLOG\_\allowbreak{}POS} 
records a vector of per-domain frontiers. The conditional bridge maps this vector to a quiet 
point when its high-minus-low difference is empty across all domains, assuming the vector 
maps to a single consumed-history prefix. 
When that difference is nonempty, tagged 3.6.1.Final source code
requests a reread. The shared helper requires the snapshot to be
running, the deduplication window to be open, and the buffer to be
populated before the reread occurs
~\cite{debezium_mariadb_readonly_source,debezium_mysql_readonly_source}.
The code therefore does not establish that every populated attempt
with a nonempty difference is retried. If the buffer is empty, there
are no copied rows to release. Under the conditions in
\S\ref{sec:inst-readonly}, processing the log still gives correct replay.
The quiet-point and empty-buffer cases do not cover every path through
the tagged implementation. A populated buffer with a nonempty
difference requires further analysis of how the sampled coordinates
bound the read and how the connector reconciles it with logged changes.
This limits what the present analysis establishes about the connector.

\paragraph{Debezium PostgreSQL}
As with MySQL's executed-GTID set, we assume that each
\texttt{pg\_\allowbreak{}current\_\allowbreak{}snapshot()} sample
represents exactly the transactions in one complete prefix of the
consumed log (A-linear, \S\ref{sec:inst-readonly}). The low sample must
include every transaction already processed by the consumer when it takes
effect. Streamed transaction IDs and sampled snapshot bounds must
share the same ID space. A PostgreSQL primary guarantees neither of the
first two conditions (\S\ref{sec:inst-readonly}). Chunk reads must be complete (A-read), row
identities must remain stable (A-scope), and log events must carry full
row images (S-IMG). The connector must also observe all required log
events (S-OBS) and follow the window-discard release rules (A-discard).

\paragraph{Shared chunk-read condition.} Official documentation for Debezium and Flink 
CDC does not fully specify the scan completeness required by A-read (\S\ref{sec:inst-readonly}, \S\ref{sec:inst-parallel}). 
Because standard transaction isolation labels like \texttt{READ COMMITTED} do not inherently 
guarantee that a query will avoid skipping rows mid-scan, our mappings treat complete chunk 
reads as an explicit engine precondition rather than inferring it from an isolation-level name.

\paragraph{Flink CDC}
While the abstract capture plan satisfies the contract independently, 
physical emitted-stream equivalence requires a correct streaming handoff. 
The connector must ensure full log retention across each chunk's window (A-window), 
faithful merging of updates and deletes at the high watermark (A-merge), complete 
log observation (S-OBS), and per-split gating that withholds live stream events at 
or before the split's high offset to prevent duplicates (A-handoff, \S\ref{sec:inst-parallel}). 
Chunk reads must additionally satisfy complete scanning (A-read). 
If a table is partitioned using a non-primary key column, that split key must remain 
immutable during capture to prevent rows from moving between splits. 
In addition, Flink's default deserializer preserves insert, update, and delete semantics 
in \texttt{RowKind}, whereas append-only output formats erase operation types and are 
covered only through an explicit normalization proof~\cite{flinkcdc_source_code}. 
Connector checkpointing, dynamic rescaling, failover, and end-to-end exactly-once 
sink application remain outside the formal scope. Modes that skip backfill reads or 
capture tables lacking primary keys are also omitted from our mappings.

\paragraph{Uncertain restores.} The bounded restore theorem assumes that the restored 
dataset represents a single engine-consistent source state at some unknown coordinate 
within recorded log bounds. It does not cover inconsistent or torn file copies assembled 
from uncoordinated reads across different times. While an exact source state exists 
at some latent coordinate within the window, the connector cannot identify it from the 
recorded evidence alone. Consequently, tracking the live source is guaranteed only 
from the known high bound onward.

\paragraph{Session-bound snapshot coordinates.}
In MariaDB point-splice capture, the snapshot status variables
(\texttt{Binlog\_\allowbreak{}snapshot\_\allowbreak{}file} and
\texttt{Binlog\_\allowbreak{}snapshot\_\allowbreak{}position}) must be queried from the exact
session holding the consistent snapshot before that view is closed or
replaced~\cite{mariadb_binlog_status_vars,mariadb_snapshot_session_source}.

\paragraph{Open reconciliation designs.}
Section~\ref{sec:merge-open} analyzes four additional reconciliation mechanisms,
namely key-range gating, version-max, re-read normalization, and rewind-by-preimage.
While their formal proof obligations are defined in our model, establishing
their full equivalence theorems remains open.

\section{Conclusion}\label{sec:conclusion}

Generalized DBLog provides a common basis for extending and combining capture
methods. Its central requirement is that copied rows and logged changes
reconstruct the selected data correctly once copying and reconciliation are
complete. The paper makes explicit which properties must come from the source
and which must be established by the connector, so a new way of obtaining a
copy can be assessed against the same contract.

We prove that the original watermarked DBLog protocol, Debezium's signal-table 
modes, Debezium's MySQL and PostgreSQL read-only protocols, and Flink CDC's
parallel-chunk protocol fulfil this contract. On a PostgreSQL primary, the
watermark-based and read-only proofs need further conditions
(\S\ref{sec:inst-classic}, \S\ref{sec:inst-readonly}).
For Debezium's MariaDB read-only protocol, the proof covers 
the quiet-point and empty-buffer cases, and a populated buffer with a nonempty
GTID difference remains open. The proofs state the source
guarantees and connector behavior required by each. Native dumps tied to exact 
log positions and backups whose log position lies within known bounds, satisfy 
the same contract.

For a read-only variant, the implementation must establish how sampled source
information identifies the changes that overlap its reads. For parallel
capture, it must preserve non-overlapping key ownership and prevent earlier
log events from reverting the reconciled data to an older state. A native dump must be paired with
the right place to continue in the log, and its keys and values must agree
with those decoded from logged changes. These checks identify what needs to be
established when the capture mechanism changes. Chunk size and parallelism can
be chosen according to source load, buffering, and retention limits while
preserving those requirements.

The composition result also allows one capture to use different methods for
different tables or key ranges. A native dump of one table can be combined
with chunked reads of another. Each method establishes its own handoff to the
same source history and their assigned keys remain non-overlapping. The
combined output then preserves the correctness guarantee. This gives
implementations room to choose a suitable copy method for each part of the
database.

The machine-checked development provides reusable proofs for these extensions 
and combinations. Delivery, crash recovery, and downstream application need to 
be verified separately. For maintainers and connector
authors, the contract provides a concrete basis for checking an existing
capture path, changing how it operates, or adding a new one while preserving
the correct treatment of copied rows, updates, and deletes.

\section*{AI assistance disclosure}

The author used Claude Fable 5, ChatGPT 5.5 and 5.6 Sol, and GPT 6 Astra
to assist with the theory, the Isabelle/HOL and Lean formalizations, the
TLA\textsuperscript{+} models, and drafting the paper's prose. The author
edited all of the prose and takes responsibility for the paper's
content and conclusions.

\bibliographystyle{unsrtnat}
\bibliography{refs}

\end{document}